\documentclass[runningheads]{llncs}

\usepackage{amsmath,amsfonts}
\usepackage{float}
\usepackage{graphicx}
\usepackage{caption}

\usepackage{algorithmic}
\usepackage{array}

\usepackage{textcomp}
\usepackage{stfloats}
\usepackage{url}
\usepackage{verbatim}
\usepackage{graphicx}
\usepackage{color, colortbl}
\usepackage[table]{xcolor}
\usepackage{enumitem}
\usepackage{amsfonts,amsmath,amssymb,graphicx,setspace}
\usepackage{framed}
\usepackage{multirow}
\usepackage{booktabs}
\usepackage{siunitx}
\usepackage{todonotes}
\usepackage{makecell}
\usepackage{longtable}
\usepackage{xspace}
\usepackage{pgfplots}
\usepackage{pgfplots, pgfplotstable}
\usepackage{bm}
\usepackage{adjustbox}
\usepackage{multicol}
\usepackage{boxedminipage}
\usepackage{multirow}
\usepackage{adjustbox}
\usepackage{multicol}

\usepackage[hidelinks]{hyperref}

\def\BibTeX{{\rm B\kern-.05em{\sc i\kern-.025em b}\kern-.08em
    T\kern-.1667em\lower.7ex\hbox{E}\kern-.125emX}}
\usepackage{balance}

\newcommand{\spc}{\ensuremath{\text{SPC}}\xspace}
\usepackage{pgfplots}

\AtBeginDocument{%
  \providecommand\BibTeX{{%
    Bib\TeX}}}

\usepackage{xspace}

\usepackage{multirow}
\usepackage[mathscr]{eucal}
\usepackage{bm}
\definecolor{Gray}{gray}{0.9}
\usepackage{adjustbox}
\usepgfplotslibrary{groupplots}
\usepgfplotslibrary{colorbrewer}

\usepackage{subcaption}

\definecolor{group1}{HTML}{1F77B4} 
\definecolor{group2}{HTML}{FF7F0E} 
\definecolor{group3}{HTML}{2CA02C} 
\definecolor{group4}{HTML}{D62728} 

\usepgfplotslibrary{colorbrewer}

\usepackage{colortbl}
\usepackage[skins]{tcolorbox}
\newtcolorbox{mybox}[2][]{%
  attach boxed title to top center
               = {yshift=-11pt},
  colframe     =black,
  colbacktitle = black,
  title        = #2,#1,
  enhanced,
}

\newcommand{\et}{\textit{et al.}\xspace}
\usepackage{amsfonts}

\newcommand{\ux}{\ensuremath{u_{\st 1}}\xspace}
\newcommand{\uy}{\ensuremath{u_{\st 2}}\xspace}
\usepackage{float}
\usepackage[section]{placeins}
\usepackage{enumitem}
\usepackage{lipsum}
\usepackage{xcolor}
 \newcommand{\st}{\scriptscriptstyle}

\usepackage{adjustbox}

\usepackage{multicol}

\usepackage{bm}
\usepackage{multicol}

\newcommand{\ses}{\ensuremath\mathtt{SS}\xspace}

\newcommand{\keygen}{\ensuremath\mathtt{KGen}\xspace}
\newcommand{\enc}{\ensuremath\mathtt{Enc}\xspace}
\newcommand{\dec}{\ensuremath\mathtt{Dec}\xspace}

\newcommand{\hadd}{\ensuremath\stackrel{\st H}+\xspace}
\newcommand{\hmul}{\ensuremath\stackrel{\st H}{\times}\xspace}

\newcommand{\se}{\ensuremath{{S}}\xspace}
\newcommand{\re}{\ensuremath{{R}}\xspace}
\newcommand{\p}{\ensuremath{{P}}\xspace}

\newcommand{\tp}{\ensuremath{{T}}\xspace}

\newcommand{\parse}{\ensuremath{\mathtt{parse}}\xspace}

\newcommand{\adv}{\ensuremath{\mathcal{A}}\xspace}
\newcommand{\simm}{\ensuremath{\mathtt{Sim}}\xspace}
\newcommand{\view}{\ensuremath{\mathtt{View}}\xspace}
\newcommand{\secsh}{\ensuremath{\mathtt{SS}^{\st(t,n)}}\xspace}
\newcommand{\empt}{\ensuremath{\epsilon}\xspace}

\usepackage{mathtools}
\usepackage{adjustbox}

\usepackage{framed}
\usepackage{tikz}
\newcommand{\h}{\ensuremath\mathtt{H}\xspace}
\newcommand{\g}{\ensuremath\mathtt{G}\xspace}

\newcommand{\ot}{\ensuremath{\mathcal{OT}^{\st 2}_{\st 1}}\xspace}
\newcommand{\dqot}{\ensuremath{\mathcal{DQ\text{--}OT}^{\st 2}_{\st 1}}\xspace}
\newcommand{\dqothf}{\ensuremath{\mathcal{DQ^{\st MR}\text{--}OT}^{\st 2}_{\st 1}}\xspace}
\newcommand{\rdqothf}{\ensuremath{{\text{DQ}^{\st \text{MR}}\text{--}\text{OT}}}\xspace}
\newcommand{\duqot}{\ensuremath{\mathcal{DUQ\text{--}OT}^{\st 2}_{\st 1}}\xspace}
\newcommand{\duqothf}{\ensuremath{\mathcal{DUQ^{\st MR}\text{--}OT}^{\st 2}_{\st 1}}\xspace}
\newcommand{\rduqothf}{\ensuremath{{\text{DUQ}^{\st \text{MR}}\text{--}\text{OT}}}\xspace}
\newcommand{\onenot}{\ensuremath{\mathcal{OT}^{\st n}_{\st 1}}\xspace}

\usetikzlibrary{arrows,decorations.markings}

\AtBeginDocument{%
  \providecommand\BibTeX{{%
    \normalfont B\kern-0.5em{\scshape i\kern-0.25em b}\kern-0.8em\TeX}}}
    
\usepackage{fullpage}

\begin{document}

\title{Delegated-Query Oblivious Transfer\\ and its Practical Applications}

\author{%
Yvo  Desmedt\thanks{y.desmedt@cs.ucl.ac.uk}\inst{1}
\hspace{1mm} and \hspace{1mm}
Aydin Abadi\thanks{aydin.abadi@newcastle.ac.uk}\inst{2} 
\institute{The University of Texas at Dallas \and Newcastle University} 
 }


\maketitle

\begin{abstract}
Databases play a pivotal role in the contemporary World Wide Web and the world of cloud computing. 
%
%
Unfortunately, numerous privacy violations have recently garnered attention in the news. 
%
%
To enhance database privacy, we consider Oblivious Transfer (OT), an elegant
cryptographic technology. Our observation reveals that existing research in this domain primarily concentrates on theoretical cryptographic applications, overlooking various practical aspects:

\vspace{.5mm}
\begin{itemize}
    \item OTs assume parties have direct access to databases. Our ``1-out-of-2 Delegated-Query OT''
enables parties to privately query a database, without direct access. 

\item With the rise of cloud computing, physically separated databases may no longer remain so. Our ``1-out-of-2 Delegated-Query Multi-Receiver OT'' protects privacy in such evolving scenarios.

\item Research often ignores the limitations of thin clients, e.g., Internet of Things devices. To address this, we propose a compiler that transforms any 1-out-of-$n$ OT into a thin client version.


\end{itemize}
\end{abstract}


\section{Introduction}\label{sec:intro}


Databases play a vital role in e-commerce, advertising, intelligence analysis, combating crime, knowledge discovery, and conducting scientific research. Privacy breaches involving databases, impacting both organizations and individuals, have become headline news. 
%
%
Some databases (e.g., Fortune 500 companies' databases about customers' purchase history)  can be valued at millions of dollars.

Privacy issues arise when users seek access to databases they do not own or create. Furthermore, many databases are now hosted in the cloud,  adding another layer of complexity to the privacy landscape.
To simultaneously protect the privacy of the user and the database itself from each other, the cryptographic-based technology, called Oblivious Transfer (OT) has been proposed. It allows a user (called a receiver) interested in the $s$-th element of a database  $(m_{\st 0}, m_{\st 1})$  (held by a sender) to learn only $m_{\st s}$ while preserving the privacy of (i) index $s\in\{0, 1\}$ from the sender and (ii) the rest of the database's elements from the receiver. Numerous variants have been developed since OT's introduction in 1981 \cite{Rabin-OT}. 

OT is an important cryptographic protocol that has found applications within various domains, such as generic secure Multi-Party Computation (MPC)  \cite{Yao82b,AsharovL0Z13,HarnikIK07}, Private Set Intersection (PSI) \cite{DongCW13}, contract signing \cite{EvenGL85}, Federated Learning (FL) \cite{YangLCT19,RenYC22,XuLZXND22}, and accessing sensitive field elements of remote private databases while preserving privacy \cite{CamenischDN09,AielloIR01,libert2021adaptive}. 
As evidenced by this work, numerous research gaps persist in this domain. Many real-world applications have been overlooked.  These oversights align with gaps in the research on OT, as expounded upon in the next section.

\clearpage

\section{Motivations and Survey}
In this section, we motivate the paper through real-world scenarios, discuss gaps in the OT research, and outline our contributions.


\subsection{Motivations}

\subsubsection{Dealing with Insiders, e.g., in Financial Institutions.}\label{sec::} 
Insider attacks pose imminent threats to various organizations and their clients, such as financial institutions and their customers.  Insiders may collaborate with external fraudsters, obtaining highly valuable data. There have been real-world incidents where bank employees have leaked or misused customers' information.

The ``Swiss Leaks'' \cite{leigh2015hsbc} is a good example to illustrate the problem of insider leaks in the banking world.  In the Swiss Leaks case, an insider attempted to sell information about accounts held by HSBC in Geneva. Later, when he failed, he leaked the information to the public.  As another example, in the case of ``JPMorgan Chase Insider Thief'' \cite{JPMorgan}, a former JPMorgan Chase personal banker has been arrested by the FBI on charges that he stole customers' account information and sold it to an undercover informant. Another notable case involved a Citibank employee who accessed sensitive customer information and used it to commit fraud \cite{citibank}. 

Additionally, a former financial advisor at Morgan Stanley, was discovered to have illicitly accessed and leaked sensitive information from approximately 730,000 accounts. This data breach compromised customers' personal details such as their names, addresses, account numbers, and investment information. This employee, who worked within Morgan Stanley's private wealth management division, entered a guilty plea in federal court in Manhattan \cite{morgan-stanley}.



In this context, an insider can exclusively target high-profile wealthy individuals and sell the victims' information to their rivals, who might make strategic investments, often remaining stealthy from the victims' perspective. For an insider, a data breach in private banking or private financial advising can be more alluring than leaking hundreds of bank accounts. Indeed, the former could yield a higher payoff while concurrently posing a lower risk of exposure. Additionally, outsiders who infiltrate the computers of an individual advisor or the third-party database can compromise the privacy of customers' queries.

Furthermore, in this setting, financial advisors, within a financial institution, frequently maintain paid subscriptions to a valuable database (e.g., containing real estate market information, market trends, and capital flows) offered by third-party providers such as CoreLogic\footnote{\url{https://www.corelogic.com/data-solutions/property-data-solutions/discovery-platform}},  Multiple Listing Service\footnote{\url{https://www.mls.com}}, or Real Capital Analytics\footnote{\url{https://www.msci.com/our-solutions/real-assets/real-capital-analytics}}.  In contrast, clients of these advisors do not necessarily need to subscribe to the database themselves. Instead, they interact with the advisors and direct their queries to them.

Hence, there is a pressing need to (i) protect customer query privacy from advisors and databases, (ii) ensure the privacy of the database from both customers and advisors, and (iii) secure the privacy of customers in the event of a data breach on the advisor's or database's side. As explained in Section \ref{sec::gap-delegated-Q-OT}, current OTs fall short of providing these features simultaneously.

\subsubsection{Multi-Receiver OT.}\label{sec:Delegated-Query-OT-w-HF}

The adoption of cloud computing has been accelerating. The ``PwC’s 2023 Cloud Business Survey'' suggests that 78\% of executives participating in the survey have mentioned that their companies had adopted cloud in most or all parts of the business \cite{PwC2023}. Moreover, multiple (sensitive) databases belonging to different parties have been merged and hosted by a single cloud provider.  Indeed, the recent cyber-attack revealed that data belonging to British Airways, Boots, BBC, and Aer Lingus was kept by the same cloud \cite{BBB-hack}. Another example is Salesforce data exposure, where a Salesforce software bug allowed users from different organizations to access each other's data within the Salesforce Marketing Cloud \cite{salesforce1}. This incident has potentially impacted various customers, including organizations like Aldo, Dunkin Donuts, Nestle Waters, and Sony \cite{salesforce2}. 
The current OTs do not allow us to deal with this scenario, as we will elaborate in Section \ref{sec::gap-Querying-Merged-Databases}. 

\subsubsection{Querying Databases with Hidden Fields.}\label{sec::delegated-Unknown-Query-OT-motivation}


In specific applications, such as finance or healthcare, sensitive details about customers or patients must be withheld from them (at least for a certain time period). In the financial sector, this may include  
(a) a binary flag that determines whether a certain customer is deemed suspicious \cite{arora2023privacy,DrivenData}, or (b) proprietary banking strategies tailored to individual clients. In the medical sector, such information may involve a binary flag indicating whether a patient has a certain condition. In certain cases, revealing specific details about an illness or test result might endanger the patient \cite{exchange-of-health-info,Withholding-Information}. 
Hence, the result that a client/receiver obtains for its request (e.g., seeking investment advice) depends on the private flag/query $s$,  provided by a third party to its advisor who is directly dealing with the client, while the client itself is not aware of the value of $s$. We will further discuss it in Section \ref{sec::Hidden-Fields}.








\subsection{Research Gaps}

\subsubsection{Support for Delegated-Query OT.}\label{sec::gap-delegated-Q-OT} 
Current OT technologies assume that a receiver that generates the query \textit{always} has \emph{direct subscription/access to databases} and enough computation resources to generate queries that are sent to the sender. This assumption has to be relaxed when receivers are not subscribed to the database (e.g., they cannot afford it) or when receivers are thin clients, e.g., IoT devices with limited computational resources or battery lifespan. We introduce \textit{Delegated-Query Oblivious Transfer} to address these limitations and deal with insider attacks (see Section \ref{sec:protocol}). 


\subsubsection{Querying Merged Databases.} \label{sec::gap-Querying-Merged-Databases}
Existing techniques do not support querying \emph{merged databases} in a  privacy-preserving manner. Specifically, they  
 are not suitable for the \emph{real-world multi-receiver setting} where a sender maintains multiple records\footnote{A database table consists of records/rows and fields/columns. Each record in a table represents a set of related data, e.g., last name and address.} each belonging to a different receiver.   The existing techniques do not allow a receiver to privately query such records without disclosing (i) the records, that the receiver accesses, to the sender and (ii) the number of records, that other database users have, to each receiver. Receivers with different levels of security form a natural example. The existing OTs reveal the entire database's size to receivers enabling them to acquire non-trivial information. The mere existence of private data can be considered sensitive information \cite{securityy}. We propose several \textit{Multi-Receiver OTs} to support querying merged databases in a privacy-preserving way (see Section \ref{sec::Multi-Receiver-OT}).

\subsubsection{Databases with Hidden Fields.} \label{sec::Hidden-Fields}
The current OT concept assumes the receiver knows the full query, which may not always be desired, as discussed in Section  \ref{sec::delegated-Unknown-Query-OT-motivation}. We will propose OT variants supporting a (partially) \textit{unknown query} (see Sections \ref{sec::DUQ-OT-protocol} and \ref{sec::DUQ-OT-HF}).

\subsubsection{Constant Size Response.}
The current techniques that allow a receiver to obtain  a response with a \emph{constant size} for its query necessitate the receiver to possess a storage space proportional to the size of the database, to locally store the database encryption. However, meeting this demanding requirement will become challenging for a thin client (e.g., in IoT settings), if its available storage space is significantly smaller than the database size. We will introduce a generic compiler that transforms any OT with a non-constant response size to one with a constant response size (see Section \ref{sec::the-compiler}).




\subsection{Our Contributions}


In this paper, we propose solutions to the aforementioned limitations using the following new techniques:

\begin{enumerate}[leftmargin=5.5mm]
\item \underline{$1$-out-of-$2$ Delegated-Query Oblivious Transfer} \underline{(\dqot)}: a new notion of OT  that extends the basic features of OT by allowing the receiver to \textit{delegate} two tasks: (i) computing the query and (ii) interacting with the sender. These tasks can be assigned to a pair of potentially semi-honest parties, $\p_{\st 1}$
and $\p_{\st 2}$, while ensuring that the sender and receiver privacy is also protected from $\p_{\st 1}$ and $\p_{\st 2}$. Section \ref{sec::sec-def} presents \dqot. 
\begin{enumerate}

\item {Delegated-Query OT (DQ-OT)}: a protocol that realizes   \dqot.  Section \ref{sec::DQ-OT} presents DQ-OT.

\item {Delegate-Unknown-Query OT (DUQ-OT)}:  a variant of DQ-OT which allows the receiver to extract the related message $m_s$ even if it does not (and must not) know the related index $s$. Section \ref{sec::DUQ-OT-protocol} presents  DUQ-OT.
\end{enumerate}

\item  \underline{$1$-out-of-$2$ Delegated-Query Multi-Receiver OT} (\dqothf): a new notion of OT that (in addition to offering OT's primary features) ensures (i) a receiver learns nothing about the total number of records and their field elements and (ii) the sender who maintains $z$ records $[(m_{\st 0, 0}, m_{\st 1, 0}),\ldots,$ $ (m_{\st 0, z-1},$ $m_{\st 1, z-1})]$ does not find out which query belongs to which record. Section \ref{sec::DQOT-HF} presents \dqothf. 

\begin{enumerate}
\item  {Delegated-Query Multi-Receiver OT (\rdqothf)}: an efficient protocol that realizes \dqothf. It is built upon DQ-OT and inherits its features. \rdqothf achieves its goal by allowing $\p_{\st 1}$ to know which record is related to which receiver.   Section \ref{sec::rdqothf} presents \rdqothf.
\item  {Delegate-Unknown-Query Multi-Receiver OT}  {(\rduqothf)}: a variant of \rdqothf which considers the case where  $\p_{\st 1}$ and $\p_{\st 2}$ do not (and must not) know which record in the database belongs to which receiver. Section \ref{sec::DUQ-OT-HF} presents \rduqothf. 
\end{enumerate}

\item  \underline{A compiler}: a generic compiler that transforms any $1$-out-of-$n$ OT that requires the receiver to receive $n$ messages (as a response) into a $1$-out-of-$n$ OT that lets a receiver (i) receive only a \emph{constant} number of messages and (ii) have constant storage space. Section \ref{sec::the-compiler} presents the compiler. 

%
%
%
%
%
%

\end{enumerate} 


\section{Preliminaries}

\subsection{Notations}\label{sec::notations}
By $\empt$ we mean an empty string. When $y$ represents a single value, $|y|$ refers to the bit length of $y$. However, when $y$ is a tuple, $|y|$ denotes the number of elements contained within $y$. We denote a sender by $\se$ and a receiver by $\re$. We assume parties interact with each other through a regular secure channel. 
We define a parse function as $\parse(\lambda, y)\rightarrow (\ux, \uy)$, which takes as input a value $\lambda$ and a value $y$ of length at least  $\lambda$-bit. It parses $y$ into two values  $\ux$ and $\uy$ and returns $(\ux, \uy)$ where the bit length of $\ux$ is $|y|-\lambda$ and the bit length of $\uy$ is  $\lambda$. 
Also, $U$ denotes a universe of messages $m_{\st 1},\ldots, m_{\st t}$. We define $\sigma$ as the maximum size of messages in $U$, i.e., $\sigma=Max(|m_{\st 1}|,\ldots, |m_{\st t}|)$. 
We use two hash functions $\h:\{0, 1\}^{\st *}\rightarrow \{0, 1\}^{\sigma}$ and $\g:\{0, 1\}^{\st *}\rightarrow \{0, 1\}^{\st\sigma+\lambda}$ modelled as random oracles \cite{Canetti97}. 



\subsection{Security Model}\label{sec::sec-model}

In this paper, we rely on the simulation-based model of secure multi-party computation \cite{DBLP:books/cu/Goldreich2004} to define and prove the proposed protocols. Below, we restate the formal security definition within this model.

 \subsubsection{Two-party Computation.} A two-party protocol $\Gamma$ problem is captured by specifying a random process that maps pairs of inputs to pairs of outputs, one for each party. Such process is referred to as a functionality denoted by  $f:\{0,1\}^{\st  *}\times\{0,1\}^{\st  *}\rightarrow\{0,1\}^{ \st *}\times\{0,1\}^{ \st *}$, where $f:=(f_{\st  1},f_{\st  2})$. For every input pair $(x,y)$, the output pair is a random variable $(f_{\st  1} (x,y), f_{\st  2} (x,y))$, such that the party with input $x$ wishes to obtain $f_{\st  1} (x,y)$ while the party with input $y$ wishes to receive $f_{\st  2} (x,y)$. 
%
 In the setting where $f$ is asymmetric and only one party (say the first one) receives the result, $f$ is defined as $f:=(f_{\st  1}(x,y), \empt)$.

 \subsubsection{Security in the Presence of Passive Adversaries.}  In the passive adversarial model, the party corrupted by such an adversary correctly follows the protocol specification. Nonetheless, the adversary obtains the internal state of the corrupted party, including the transcript of all the messages received, and tries to use this to learn information that should remain private. Loosely speaking, a protocol is secure if whatever can be computed by a party in the protocol can be computed using its input and output only. In the simulation-based model, it is required that a party’s view in a protocol's 
 execution can be simulated given only its input and output. This implies that the parties learn nothing from the protocol's execution. More formally, party $i$’s view (during the execution of $\Gamma$) on input pair  $(x, y)$ is denoted by $\mathsf{View}_{\st  i}^{\st  \Gamma}(x,y)$ and equals $(w, r_{\st  i}, m_{\st  1}^{\st  i}, \ldots, m_{\st  t}^{\st  i})$, where $w\in\{x,y\}$ is the input of $i^{\st  th}$ party, $r_{\st  i}$ is the outcome of this party's internal random coin tosses, and $m_{\st  j}^{\st  i}$ represents the $j^{\st  th}$ message this party receives.  The output of the $i^{\st  th}$ party during the execution of $\Gamma$ on $(x, y)$ is denoted by $\mathsf{Output}_{\st  i}^{\st  \Gamma}(x,y)$ and can be generated from its own view of the execution.  
\vspace{-1mm}
\begin{definition}
Let $f$ be the deterministic functionality defined above. Protocol $\Gamma$ securely computes $f$ in the presence of a  passive adversary if there exist polynomial-time algorithms $(\mathsf {Sim}_{\st  1}, \mathsf {Sim}_{\st  2})$ such that:
\end{definition}

  \begin{equation*}
  \{\mathsf {Sim}_{\st 1}(x,f_{\st 1}(x,y))\}_{\st x,y}\stackrel{c}{\equiv} \{\mathsf{View}_{\st 1}^{\st \Gamma}(x,y) \}_{\st x,y}
  \end{equation*}
  \begin{equation*}
    \{\mathsf{Sim}_{\st 2}(y,f_{\st 2}(x,y))\}_{ \st x,y}\stackrel{c}{\equiv} \{\mathsf{View}_{\st 2}^{\st \Gamma}(x,y) \}_{\st x,y}
  \end{equation*}


  \subsection{Random Permutation}
A random permutation $\pi(e_{\st 0} ,\ldots, e_{\st n})\rightarrow(e'_{\st 0},\ldots, e'_{\st n})$ is a probabilistic function that takes a set $A=\{e_{\st 0}, \ldots, e_{\st n}\}$ and returns the same set of elements in a permuted order $B=\{e'_{\st 0},\ldots, e'_{\st n}\}$. The security of $\pi(.)$ requires that given set $B$ the probability that one can find the original index of an element $e'_{\st i}\in B$ is $\frac{1}{n}$.  In practice, the Fisher-Yates shuffle algorithm \cite{Knuth81} can permute a set of $n$ elements in time $O(n)$. We will use $\pi(.)$ in the protocols presented in Figures \ref{fig::DQ-OT-with-unknown-query} and \ref{fig::DQ-OT-with-unknown-query-and-HF-first-eight-phases}.



\subsection{Diffie-Hellman Assumption}\label{sec::dh-assumption}
Let $G$ be a group-generator scheme, which on input $1^{\st\lambda}$ outputs $(\mathbb{G}, p, g)$ where $\mathbb{G}$ is the description of a group, $p$ is the order of the group which is always a prime number, $\log_{\st 2}(p)=\lambda$ is a security parameter and $g$ is a generator of the group. In this paper, $g$ and $p$ can be selected by sender \se (in the context of OT). 

\subsubsection*{Computational Diffie-Hellman (CDH) Assumption.} 
We say that $G$ is hard under CDH assumption, if for any probabilistic polynomial time (PPT) adversary $\mathcal{A}$, given $(g^{\st a_1}, g^{\st a_2})$ it has only negligible probability to correctly compute $g^{\st a_1\cdot a_2}$. More formally, it holds that $Pr[\mathcal{A}(\mathbb{G}, p, g, g^{\st a_1}, g^{\st a_2})\rightarrow g^{\st a_1\cdot a_2}]\leq \mu(\lambda)$, where $(\mathbb{G}, p, g) \stackrel{\st\$}\leftarrow G(1^{\st\lambda})$, $a_{\st 1}, a_{\st 2} \stackrel{\st\$}\leftarrow \mathbb{Z}_{\st p}$, and $\mu$ is a negligible function  \cite{DiffieH76}.

\subsection{Secret Sharing}\label{sec::secret-haring}

A (threshold) secret sharing $\mathtt{SS}^{\st(t,n)}$  scheme is a cryptographic protocol that enables a dealer to distribute a string $s$, known as the secret, among $n$ parties in a way that the secret $s$ can be recovered when at least a predefined number of shares, say $t$, are combined. If the number of shares in any subset is less than $t$, the secret remains unrecoverable and the shares divulge no information about $s$. This type of scheme is referred to as $(n, t)$-secret sharing or \secsh for brevity. 

In the case where $t=n$, there exists a highly efficient XOR-based secret sharing \cite{blakley1980one}. In this case, to share the secret $s$, the dealer first picks $n-1$ random bit strings $r_{\st 1}, \ldots, r_{\st n-1}$ of the same length as the secret. Then, it computes $r_{\st n} = r_{\st 1} \oplus \ldots \oplus  r_{\st n} \oplus s$. It considers each $r_{\st i}\in\{r_{\st 1}, \ldots,r_{\st n}\}$ as a share of the secret. To reconstruct the secret, one can easily compute $r_{\st 1}\oplus \ldots \oplus r_{\st n}$. Any subset of less than $n$ shares reveals no information about the secret. We will use this scheme in this paper. A secret sharing scheme involves two main algorithms; namely, $\ses(1^{\st \lambda}, s, n, t)\rightarrow (r_{\st 1}, \ldots, r_{\st n})$: to share a secret and $\mathtt{RE}(r_{\st 1}, \ldots, r_{\st t}, n, t)\rightarrow s$ to reconstruct the secret. 

\subsection{Additive Homomorphic Encryption}\label{sec::AHE}

Additive homomorphic encryption involves three algorithms: (1) key generation: $\keygen(1^{\st\lambda})\rightarrow (sk, pk)$, which takes a security parameter as input and outputs a secret and public keys pair, (2) encryption: $\enc(pk, m)\rightarrow c$, that takes public key $pk$ and a plaintext message $m$ as input and returns a ciphertext $c$, and (3) decryption: $\dec(sk, c)\rightarrow m$, which takes secret key $sk$ and ciphertext $c$ as input and returns plaintext message $m$.  It has the following properties: 
 
 \begin{itemize}
 \item [$\bullet$]  Given two ciphertexts $\enc(pk, m_{\st 1})$ and $\enc(pk, m_{\st 2})$, one can compute the encryption of the sum of related plaintexts: 
  $\dec(sk,\mathtt{Enc}(pk, m_{\st 1})\hadd \enc(pk, m_{\st 2}))= m_{\st 1}+m_{\st 2}$, where $\hadd$ denotes homomorphic addition.

  
 \item [$\bullet$] Given a ciphertext $\mathtt{Enc}(pk, m)$ and a plaintext message $c$, one can compute the encryption of the product of related plaintexts: 
 %
 %
  $\dec(sk, \enc(pk, m)\hmul c) = m\cdot c$, where $\hmul$ denotes homomorphic multiplication.
 \end{itemize}
We require that the encryption scheme satisfies indistinguishability against chosen-plaintext attacks (IND-CPA).  We refer readers to \cite{KatzLindell2014} for a formal definition. One such scheme that meets the above features is the Paillier public key cryptosystem, proposed in \cite{Paillier99}.


\section{Related Work}\label{sec::related-work}


 Oblivious Transfer (OT) is one of the vital building blocks of cryptographic protocols and has been used in various mechanisms, such as PSI, generic MPC, and zero-knowledge proofs. 
 The traditional $1$-out-of-$2$ OT (\ot) is a protocol that involves two parties, a sender \se and a receiver \re.  \se has a pair of input messages $(m_{\st 0}, m_{\st 1})$ and \re has an index $s$. The aim of \ot is to allow \re to obtain $m_{\st s}$, without revealing anything about $s$ to \se, and without allowing \re to learn anything about  $m_{\st 1-s}$. The traditional \ot functionality is defined as $\mathcal{F}_{\st\ot}:((m_{\st 0}, m_{\st 1}), s) \rightarrow (\empt, m_{\st s})$.   
  
The notion of $1$-out-of-$2$ OT was initially proposed by Rabin \cite{Rabin-OT} which consequently was generalized by Even \et \cite{EvenGL85}. Since then, numerous variants of OT have been proposed. For instance, (i) $1$-out-of-$n$ OT, e.g., in \cite{NaorP99,Tzeng02,LiuH19}: which allows \re to pick one entry out of $n$ entries held by \se, (ii) $k$-out-of-$n$ OT, e.g., in \cite{ChuT05,JareckiL09,ChenCH10}: which allows \re to pick $k$ entries out of $n$ entries held by \se, (iii) OT extension, e.g., in \cite{IshaiKNP03,Henecka013,Nielsen07,AsharovL0Z13}: that supports efficient executions of OT (which mainly relies on symmetric-key operations), in the case OT needs to be invoked many times, and (iv) distributed OT, e.g., in \cite{NaorP00,CorniauxG13,ZhaoSJMZX20}: that allows the database to be distributed among $m$ servers/senders. 

In the remainder of this section, we discuss several variants of OT  that have extended and enhanced the original OT in \cite{Rabin-OT}. 

\subsection{Distributed OT} 

 Naor and Pinkas \cite{NaorP00} proposed several protocols for distributed OT where the role of sender \se (in the original OT) is divided between several servers.  In these schemes, a receiver must contact a threshold of the servers to run the OT. 
 
 The proposed protocols are in the semi-honest model. They use symmetric-key primitives and do not involve any modular exponentiation that can lead to efficient implementations.  These protocols are based on various variants of polynomials (e.g., sparse and bivariate), polynomial evaluation, and pseudorandom function. In these distributed OTs, the security against the servers holds as long as less than a predefined number of these servers collude. 
 Later, various distributed OTs have been proposed\footnote{Distributed OT has also been called proxy OT in \cite{YaoF06}.}. For instance, Corniaux and Ghodosi \cite{CorniauxG13} proposed a verifiable $1$-out-of-$n$ distributed OT that considers the case where a threshold of the servers are potentially active adversaries. The scheme is based on a sparse $n$-variate polynomial, verifiable secret sharing, and error-correcting codes. 
 
 Moreover, Zhao \et  \cite{ZhaoSJMZX20} proposed a distributed version of OT extension that aims to preserve the efficiency of OT extension while delegating the role of \se to multiple servers a threshold of which can be potentially semi-honest. The scheme is based on a hash function and an oblivious pseudorandom function.  
However, there exists no OT that supports the delegation of the query computation to third-party servers in a privacy-preserving manner. 


\subsection{Multi-Receiver OT}\label{sec::multi-rec-ot}

Camenisch \et \cite{CamenischDN09} proposed a protocol for ``OT with access control''. It involves a set of receivers and a sender which maintains records of the receivers. It offers a set of interesting features; namely, (i) only authorized receivers can access certain records; (ii) the sender does not learn which record a receiver accesses, and (iii) the sender does not learn which roles (or security clearance) the receiver has when it accesses the records. 
In this scheme, during the setup, the sender encrypts all records (along with their field elements) and publishes the encrypted database for the receivers to download. Subsequently, researchers proposed various variants of OT with access control, as seen in \cite{CamenischDNZ11,CamenischDEN12,CamenischDN10}. 
Nevertheless, in all the aforementioned schemes, the size of the entire database is revealed to the receivers.


\subsection{OT with Constant Response Size}

Researchers have proposed several OTs, e.g., those proposed in \cite{CamenischNS07,GreenH08,ZhangLWR13},  that enable a receiver to obtain a constant-size response to its query. To achieve this level of communication efficiency, these protocols require the receiver to locally store the encryption of the entire database, in the initialization phase. During the transfer phase, the sender assists the receiver with locally decrypting the message that the receiver is interested in. 
The main limitation of these protocols is that a thin client with limited available storage space cannot use them, as it cannot locally store the encryption of the entire database.  We refer readers to \cite{OT-Survey} for a recent survey of OT.

%
%
%
%

%

\section{Delegated-Query OT}\label{sec:protocol}
In this section, we present the notion of Delegated-Query $1$-out-of-$2$ OT (\dqot) and a protocol that realizes it. 
 \dqot involves four parties; namely, sender \se, receiver \re, and two helper servers $\p_{\st 1}$ and $\p_{\st 2}$ that assist \re in computing the query.  
 \dqot enables \re to delegate (i) the computation of the query and (ii) the interaction with \se to $\p_{\st 1}$ and $\p_{\st 2}$, who jointly compute \re's query and send it to \se.  
\dqot (in addition to offering the basic security of OT) ensures that \re's privacy is preserved from  $\p_{\st 1}$ and $\p_{\st 2}$, in the sense that  $\p_{\st 1}$ and $\p_{\st 2}$ do not learn anything about the actual index (i.e., $s\in\{0,1\}$) that \re is interested in, if they do not collude with each other.

\subsection{Functionality Definition}

Informally, the functionality that \dqot computes takes as input (i) a pair of messages $(m_{\st 0}, m_{\st 1})$ from \se, (ii) empty string  \empt from $\p_{\st 1}$, (iii)  empty string  \empt from $\p_{\st 2}$, and (iv) the index  $s$ (where $s\in \{0, 1\}$) from \re. It outputs an empty string $\empt$ to  \se, $\p_{\st 1}$, and $\p_{\st 2}$, and outputs the message with index $s$, i.e., $m_{\st s}$, to \re.  Formally, we define the functionality as: 
 $\mathcal{F}_{\scriptscriptstyle\dqot}:\big((m_{\st 0}, m_{\st 1}), \empt, \empt, s\big) \rightarrow (\empt, \empt, \empt, m_{\st s})$.

 \subsection{Security Definition}\label{sec::sec-def}
 
 Next, we present a formal definition of \dqot.

\begin{definition}[\dqot]\label{def::DQ-OT-sec-def} Let $\mathcal{F}_{\scriptscriptstyle\dqot}$ be the delegated-query OT functionality defined above. We say protocol $\Gamma$ realizes $\mathcal{F}_{\scriptscriptstyle\dqot}$ in the presence of passive adversary \se, \re, $\p_{\st 1}$, or $\p_{\st 2}$ if for  every non-uniform PPT adversary \adv
in the real model, there exists a non-uniform PPT adversary (or simulator) \simm  in
the ideal model, such that:
%
\begin{equation}\label{equ::sender-sim-}
\begin{split}
\Big\{\simm_{\st\se}\big((m_{\st 0}, m_{\st 1}), \empt\big)\Big\}_{\st m_{\st 0}, m_{\st 1}, s}\stackrel{c}{\equiv} \Big\{\view_{\st\se}^{\st \Gamma}\big((m_{\st 0}, m_{\st 1}), \empt,  \empt, s\big) \Big\}_{ m_{\st 0}, m_{\st 1}, s}
\end{split}
\end{equation}

\begin{equation}\label{equ::server-sim-}
\begin{split}
\Big\{\simm_{\st\p_i}(\empt, \empt)\Big\}_{\st m_{\st 0}, m_{\st 1}, s}\stackrel{c}{\equiv}     \Big\{\view_{\st\p_i}^{\st \Gamma}\big((m_{\st 0}, m_{\st 1}), \empt,  \empt, s\big) \Big\}_{\st\st m_{\st 0}, m_{\st 1}, s}
\end{split}
\end{equation}

\begin{equation}\label{equ::reciever-sim-}
\begin{split}
\Big\{\simm_{\st\re}\Big(s, \mathcal{F}_{\scriptscriptstyle\dqot}\big((m_{\st 0}, m_{\st 1}), \empt,  \empt, s\big)\Big)\Big\}_{\st m_{\st 0}, m_{\st 1}, s}\stackrel{c}{\equiv}  \Big\{\view_{\st\re}^{\st \Gamma}\big((m_{\st 0}, m_{\st 1}), \empt,  \empt, s\big) \Big\}_{\st m_{\st 0}, m_{\st 1}, s}
\end{split}
\end{equation}
for all $i$,  $i\in \{1,2\}$.

\end{definition}
Intuitively, Relation \ref{equ::sender-sim-} states that the view of a corrupt \se during the execution of protocol $\Gamma$ (in the real model) can be simulated by a simulator $\simm_{\st\se}$ (in the ideal model) given only \se's input and output, i.e., $(m_{\st 0}, m_{\st 1})$ and \empt respectively. 

Relation \ref{equ::server-sim-} states that the view of each corrupt server $\p_{\st i}$ during the execution of  $\Gamma$ can be simulated by a simulator $\simm_{\st\p_i}$ given only $\p_{\st i}$'s input and output, i.e., $\empt$ and $\empt$ respectively. 

Relation \ref{equ::reciever-sim-} states that the view of a corrupt \re during the execution of $\Gamma$ can be simulated by a simulator $\simm_{\st\re}$ given only \re's input and output, i.e., $s$ and $m_{\st s}$ respectively.

A \dqot scheme must meet two (new) properties, \textit{efficiency} and \textit{sender-push communication} (\spc). 
Efficiency states that the query generation of the receiver is faster compared to traditional (non-delegated) OT. 
\spc, on the other hand, stipulates that the sender transmits responses to the receiver without requiring the receiver to directly initiate a query to the sender. Below, we formally state these properties. 
%




\begin{definition}[Efficiency]\label{def::efficiency}
A \dqot scheme is considered efficient if the running time of the receiver-side request (or query) generation algorithm, denoted as $\mathtt{Request(1^{\st \lambda}, s, pk)}$, satisfies two conditions: 

\begin{enumerate}[label=$\bullet$,leftmargin=4.7mm]
    \item The running time is upper-bounded by $poly(|{m}|)$, where $poly$ is a fixed polynomial, $m$ is a tuple of messages that the sender holds, and $|{m}|$ represents the number of elements/messages in tuple $m$.

    \item The running time is asymptotically constant with respect to the security parameter $\lambda$, i.e., it is $O(1)$.  
\end{enumerate}
\end{definition}


\begin{definition}[Sender-push communication]\label{def::Server-push-comm} Let 
(a) $\mathtt{Action}_{\st \re}(t)$ represent the set of actions available to \re at time $t$ (these actions may include sending requests, receiving messages, or any other interactions \re can perform within the scheme's execution), 
(b) $\mathtt{Action}_{\st \se}(t)$ be the set of actions available to \se at time $t$,  
(c) $\mathtt{SendRequest}(\re,\se)$ be the action of \re sending a request to \se, and 
(d) $\mathtt{SendMessage}(\se,\re)$ represents the action of \se sending a message to \re. Then, we say that a \dqot scheme supports sender-push communication, if it meets the following conditions:

\begin{enumerate}[label=$\bullet$,leftmargin=4.7mm]

\item \textit{Receiver-side restricted interaction}: 
For all $t$ in the communication timeline (i.e., within the execution of the scheme), the set of actions $\mathtt{Action}_{\st \re}(t)$ available to the receiver \re is restricted such that it does not include direct requests to the sender \se. Formally,
$$\forall t: \mathtt{Action}_{\st \re}(t) \cap \{\mathtt{SendRequest}(\re,\se)\}=\emptyset$$

\item \textit{Sender-side non-restricted interaction}: 
For all $t$ in the communication timeline, the sender \se has the capability to push messages to the receiver \re without receiving explicit request directly from \re. Formally, 
$$\forall t:
\mathtt{Action}_{\st \se}(t) \subseteq \{\mathtt{SendMessage}(\se,\re)\}
$$

\end{enumerate}

\end{definition}

Looking ahead, the above two properties are also valid for the variants of \dqot; namely, \duqot, \dqothf, and  \duqothf, with a minor difference being that the receiver-side request generation algorithm in these three variants is denoted as $\mathtt{\re.Request}()$.

\subsection{Protocol}\label{sec::DQ-OT}
Now, we present an efficient 1-out-of-2 OT protocol, called DQ-OT, that realizes \dqot. 
%
We build DQ-OT upon the \ot proposed by  Naor and Pinkas  \cite[pp. 450, 451]{Efficient-OT-Naor}. Our motivation for this choice is primarily didactic. Appendix \ref{sec::OT-of-Naor-Pinkas} restates this OT.

The high-level idea behind the design of DQ-OT is that \re splits its index into two shares and sends each share to each $\p_{\st i}$. Each $\p_{\st i}$ computes a (partial) query and sends the result to  \se which generates the response for \re in the same manner as the original OT in \cite{Efficient-OT-Naor}. 
%
%
Below, we explain how  DQ-OT operates, followed by an explanation of how it achieves  correctness.

First, \re splits the index that it is interested in into two binary shares,  $(s_{\st 1}, s_{\st 2})$. Then, it picks two random values, $(r_{\st 1}, r_{\st 2})$, and then sends each pair $(s_{\st i}, r_{\st i})$ to each $\p_{\st i}$. 

Second, to compute a partial query, $\p_{\st 2}$ treats $s_{\st 2}$ as the main index that \re is interested in and computes a partial query, $\delta_{\st s_2}=g^{\st r_2}$. Also, $\p_{\st 2}$  generates another query, $\delta_{\st 1-s_2}=\frac{C}{g^{\st r_2}}$, where $C$ is a random public parameter (as defined in \cite{Efficient-OT-Naor}). $\p_{\st 2}$ sorts the two queries in ascending order based on the value of $s_{\st 2}$ and sends the resulting $(\delta_{\st 0}, \delta_{\st 1})$ to $\p_{\st 1}$.

Third, to compute its queries, $\p_{\st 1}$ treats $\delta_{\st 0}$ as the main index (that \re is interested) and computes $\beta_{\st s_1}=\delta_{\st 0}\cdot g^{\st r_1}$. Additionally, it generates another query $\beta_{\st 1-s_1}=\frac{\delta_{\st 1}}{g^{\st r_1}}$. Subsequently, $\p_{\st 1}$ sorts the two queries in ascending order based on the value of ${s_{\st 1}}$ and sends the resulting $(\beta_{\st 0}, \beta_{\st 1})$ to $\re$. 

Fourth, given the queries, \se computes the response in the same manner it does in the original OT in \cite{Efficient-OT-Naor} and sends the result to \re who extracts from it, the message that it asked for, with the help of $s_i$ and $r_i$ values.  
The detailed  DQ-OT  is presented in Figure \ref{fig::DQ-OT}.

\begin{figure}[!htbp]
\setlength{\fboxsep}{.8pt}
\begin{center}
    \begin{tcolorbox}[enhanced,
    drop fuzzy shadow southwest,
    colframe=black,colback=white]
\begin{enumerate}

\item \underline{\textit{$\se$-side Initialization:}} 
$\mathtt{Init}(1^{\st \lambda})\rightarrow pk$
\begin{enumerate}

\item chooses a sufficiently large prime number $p$.

\item selects random element
$C\stackrel{\st \$}\leftarrow \mathbb{Z}_p$ and generator $g$.
\item publishes $pk=(C, p, g)$. 

\end{enumerate}

\item \underline{\textit{\re-side Delegation:}}
$\mathtt{Request}(1^{\st \lambda}, s, pk)\hspace{-.6mm}\rightarrow\hspace{-.6mm} req=(req_{\st 1}, req_{\st 2})$

\begin{enumerate}

\item splits  the private index $s$ into two shares $(s_{\st 1}, s_{\st 2})$ by calling  $\ses(1^{\st \lambda}, s, 2, 2)\rightarrow (s_{\st 1}, s_{\st 2})$.

\item picks two uniformly random values: $r_{\st 1}, r_{\st 2} \stackrel{\st\$}\leftarrow\mathbb{Z}_{\st p}$.

\item sends $req_{\st 1}=(s_{\st 1}, r_{\st 1})$ to $\p_{\st 1}$ and $req_{\st 2}=(s_{\st 2}, r_{\st 2})$ to $\p_{\st 2}$.

\end{enumerate}

\item \underline{\textit{$\p_{\st 2}$-side Query Generation:}}
$\mathtt{\p_{\st 2}.GenQuery}(req_{\st 2}, ,pk)\rightarrow q_{\st 2}$

\begin{enumerate}

\item computes a pair of partial queries: 
%

  $$\delta_{\st s_2}= g^{\st r_2},\ \  \delta_{\st 1-s_{\st 2}} = \frac{C}{g^{\st r_2}}$$
  
\item sends $q_{\st 2}=(\delta_{\st 0}, \delta_{\st 1})$ to  $\p_{\st 1}$. 

\end{enumerate}

\item\underline{\textit{$\p_{\st 1}$-side Query Generation:}}
$\mathtt{\p_{\st 1}.GenQuery}(req_{\st 1}, q_{\st 2},pk)\rightarrow q_{\st 1}$

\begin{enumerate}

\item computes a pair of final queries as:  
%

$$\beta_{s_{\st 1}}=\delta_{\st 0}\cdot g^{\st r_{\st 1}},\ \ \beta_{\st 1-s_1}=\frac{\delta_{\st 1}} {g^{\st r_{\st 1}}}$$

\item sends $q_{\st 1}=(\beta_{\st 0}, \beta_{\st 1})$ to  $\se$.

\end{enumerate}

\item\underline{\textit{\se-side Response Generation:}}  
$\mathtt{GenRes}(m_{\st 0}, m_{\st 1}, pk, q_{\st 1})\rightarrow res$
\begin{enumerate}

\item aborts if  $C \neq \beta_{\st 0}\cdot \beta_{\st 1}$.
\item picks two uniformly random values: 
$y_{\st 0}, y_{\st 1}  \stackrel{\st\$}\leftarrow\mathbb{Z}_{\st p}$.
\item computes a response pair $(e_{\st 0}, e_{\st 1})$ as follows:

 $$e_{\st 0} := (e_{\st 0,0}, e_{\st 0,1}) = (g^{\st y_0}, \h(\beta_{\st 0}^{\st y_0}) \oplus m_{\st 0})$$
$$e_{\st 1} := (e_{\st 1,0}, e_{\st 1,1}) = (g^{\st y_1}, \h(\beta_{\st 1}^{\st y_1}) \oplus m_{\st 1})$$
\item sends $res=(e_{\st 0}, e_{\st 1})$ to \re. 

\end{enumerate}

\item\underline{\textit{\re-side Message Extraction:}} 
$\mathtt{Retreive}(res, req, pk)\hspace{-.6mm}\rightarrow\hspace{-.8mm} m_{\st s}$
\begin{enumerate}

\item sets $x=r_{\st 2}+r_{\st 1}\cdot(-1)^{\st s_{\st 2}}$

\item retrieves the related message: 
$m_{\st s}=\h((e_{\st s, 0})^{\st x})\oplus e_{\st s, 1}$

\end{enumerate}

\end{enumerate}
\end{tcolorbox}
\end{center}

\caption{DQ-OT: Our $1$-out-of-$2$ OT that supports query delegation. The input of \re is a private binary index $s$ and the input of \se is a pair of messages $(m_{\st 0}, m_{\st 1})$. Note, $\ses(.)$ is the share-generation algorithm, $\h(.)$ is a hash function, and $\$$ denotes picking a value uniformly at random.} 

\label{fig::DQ-OT}
\end{figure}

\begin{theorem}\label{theo::DQ-OT-sec}
Let $\mathcal{F}_{\scriptscriptstyle\dqot}$ be the functionality defined in Section \ref{sec::sec-def}. If  
Discrete Logarithm (DL), Computational Diffie-Hellman (CDH), and Random Oracle (RO) assumptions hold, then DQ-OT (presented in Figure \ref{fig::DQ-OT}) securely computes $\mathcal{F}_{\scriptscriptstyle\dqot}$ in the presence of semi-honest adversaries,
w.r.t. Definition \ref{def::DQ-OT-sec-def}. 
\end{theorem}




\subsection{DQ-OT's Security Proof}\label{sec::DQ-OT-proof}
Below, we prove  DQ-OT's security, i.e., Theorem \ref{theo::DQ-OT-sec}.

\begin{proof}
We consider the case where each party is corrupt, at a time.

\subsubsection{Corrupt Receiver \re.} In the real execution, \re's view is:  
$\view_{\re}^{\st DQ\text{-}OT}\big(m_{\st 0}, m_{\st 1}, \empt,$ $  \empt, s\big) = \{r_{\st\re}, g, C, p, e_{\st 0}, e_{\st 1}, m_{\st s} \}$, where $g$ is a random generator, $C=g^{\st a}$ is a random public parameter, $a$ is a random value, $p$ is a large random prime number, and $r_{\st \re}$ is the outcome of the internal random coin of  \re and is used to generate $(r_{\st 1}, r_{\st 2})$.  
 Below, we construct an idea-model simulator  $\simm_{\st\re}$ which receives $(s, m_{\st s})$ from \re. 
 
 \begin{enumerate}
 \item initiates an empty view and appends uniformly random coin $r'_{\st\re}$ to it, where $r'_{\st\re}$ will be used to generate \re-side randomness. 
It chooses a large random prime number $p$ and a random generator $g$.


 %
 \item sets $(e'_{\st 0}, e'_{\st 1})$ as follows: 
 \begin{itemize}
 \item splits $s$ into two shares: $\ses(1^{\st\lambda}, s, 2, 2)\rightarrow (s'_{\st 1}, s'_{\st 2})$.
 \item picks uniformly random values: $C', r'_{\st 1}, r'_{\st 2},y'_{\st 0}, y'_{\st 1}\stackrel{\st\$}\leftarrow\mathbb{Z}_{\st p}$.
 %
 %
 \item sets $\beta'_s=g^{\st x}$, where $x$ is set as follows: 
 \begin{itemize}
 \item[$*$] $x = r'_{\st 2} + r'_{\st 1}$, if $(s = s_{\st 1} = s_{\st 2} = 0)$ or $(s = s_{\st 1} = 1\wedge s_{\st 2} = 0)$.
 \item[$*$] $x = r'_{\st 2} - r'_{\st 1}$, if $(s =0 \wedge s_{\st 1} = s_{\st 2} = 1)$ or $(s = s_{\st 2} = 1\wedge s_{\st 1} = 0)$.
 \end{itemize}
 \item picks a uniformly random value $u\stackrel{\st\$}\leftarrow\mathbb{Z}_{p}$ and then sets $e'_{\st s} = (g^{\st y'_s}, \h(\beta'^{\st y'_s}_{\st s})\oplus m_{\st s})$ and $e'_{\st 1-s} = (g^{\st y'_{\st 1-s}}, u)$. 
 \end{itemize}
 \item  appends $(g, C', p, r'_{\st 1}, r'_{\st 2},  e'_{\st 0}, e'_{\st 1}, m_{\st s})$ to the view and outputs the view.
 \end{enumerate}
 
 Now we discuss why the two views in the ideal and real models are indistinguishable. Since we are in the semi-honest model, the adversary picks its randomness according to the protocol description; thus, $r_{\st \re}$ and $r'_{\st \re}$ model have identical distributions, so do values $(r_{\st 1}, r_{\st 2})$ in the real model and $(r'_{\st 1}, r'_{\st 2})$ in the ideal model. Also, $C$ and $C'$ have been picked uniformly at random and have identical distributions. The same applies to values $g$ and $p$ in the real and ideal models.

 Next, we argue that  $e_{\st 1-s}$ in the real model and $e'_{\st 1-s}$ in the ideal model are indistinguishable. In the real model, it holds that $e_{\st 1-s}=(g^{\st y_{\st 1-s}}, \h(\beta^{y_{1-s}}_{\st 1-s})\oplus m_{\st 1-s})$, where $\beta^{\st y_{\st 1-s}}_{\st 1-s}=\frac{C}{g^{\st x}}=g^{\st a-x}$. Since $y_{\st 1-s}$ in the real model and $y'_{\st 1-s}$ in the ideal model have been picked uniformly at random and unknown to the adversary/distinguisher, $g^{\st y_{\st 1-s}}$ and $g^{\st y'_{\st 1-s}}$  have identical distributions.

 Furthermore, in the real model, given $C=g^{\st a}$, due to the DL problem, $a$ cannot be computed by a PPT adversary. Also, due to CDH assumption, \re cannot compute  $\beta^{y_{1-s}}_{1-s}$ (i.e., the input of $\h(.)$), given $g^{\st y_{\st 1-s}}$ and $g^{\st a-x}$. We also know that $\h(.)$ is modeled as a random oracle and its output is indistinguishable from a random value. Thus,  $\h(\beta^{\st y_{\st 1-s}}_{\st 1-s})\oplus m_{\st 1-s}$ in the real model and $u$ in the ideal model are indistinguishable. This means that $e_{\st 1-s}$ and $e'_{\st 1-s}$ are indistinguishable too, due to DL, CDH, and RO assumptions. Also, in the real and idea models, $e_{\st s}$ and $e'_{\st s}$ have been defined over $\mathbb{Z}_{\st p}$ and their decryption always result in the same value $m_{\st s}$. Thus, $e_{\st s}$ and $e'_{\st s}$ have identical distributions too. 
Also, $m_{\st s}$ has identical distribution in both models. 

We conclude that the two views are computationally indistinguishable, i.e., Relation \ref{equ::reciever-sim-} (in Section \ref{sec::sec-def}) holds.

\subsubsection{Corrupt Sender \se.} In the real model, \se's view is: 
$\view_{\se}^{\st DQ\text{-}OT}\big((m_{\st 0}, m_{\st 1}), \empt,  \empt, $ $s\big)= \{r_{\st\se}, C, \beta_{\st 0}, \beta_{\st 1}\}$, where $r_{\st \se}$ is the outcome of the internal random coin of  \se.  Next, we construct an idea-model simulator  $\simm_{\st\se}$ which receives $(m_{\st 0}, m_{\st 1})$ from \se. 

\begin{enumerate}
\item initiates an empty view and appends uniformly random coin $r'_{\st\se}$ to it, where $r'_{\st\se}$ will be used to generate random values for \se.
%

\item picks random values $C', r'\stackrel{\st\$}\leftarrow\mathbb{Z}_{\st p}$.

\item sets $\beta'_{\st 0}=g^{\st r'}$ and $\beta'_{\st 1}=\frac{C'}{g^{\st r'}}$. 

\item appends $\beta'_{\st 0}$ and $\beta'_{\st 1}$ to the view and outputs the view. 

\end{enumerate}

Next, we explain why the two views in the ideal and real models are indistinguishable. Recall, in the real model, $(\beta_{\st s}, \beta_{\st 1-s})$ have the following form: $\beta_{s}=g^{\st x}$ and  $\beta_{\st 1-s}=g^{\st a-x}$, where $a=DL(C)$ and $C=g^{\st a}$. In this model, because $a$ and $x$ have been picked uniformly at random and unknown to the adversary, due to DL assumption,  $\beta_{\st s}$ and $\beta_{\st 1-s}$ have identical distributions and are indistinguishable. In the ideal model, $r'$ has been picked uniformly at random and we know that $a'$ in $C' = g^{\st a'}$ is a uniformly random value, unknown to the adversary; therefore, due to DL assumption, $\beta'_{\st 0}$ and $\beta'_{\st 1}$ have identical distributions too. Moreover, values $\beta_{\st s}, \beta_{\st 1-s}, \beta'_{\st 0}$, and $\beta'_{\st 1}$  have been defined over the same field, $\mathbb{Z}_{\st p}$. Thus, they have identical distributions and are indistinguishable. 

Therefore, the two views are computationally indistinguishable, i.e., Relation \ref{equ::sender-sim-} (in Section \ref{sec::sec-def}) holds.

\subsubsection{Corrupt Server $\p_{\st 2}$.}  
 
 In the real execution, $\p_{\st 2}$'s view is: 
$\view_{\p_2}^{\st DQ\text{-}OT}\big((m_{\st 0}, m_{\st 1}), \empt,$ $ \empt, s\big)=\{g, C, p, s_{\st 2}, r_{\st 2}\}$. Below, we show how an ideal-model simulator $\simm_{\st \p_{\st 2}}$ works. 

\begin{enumerate}
\item initiates an empty view. It selects a random generator $g$ and a large random prime number $p$.
\item picks two uniformly random values $s'_{\st 2}\stackrel{\st\$}\leftarrow\mathbb U$ and $C', r'_{\st 2}\stackrel{\st\$}\leftarrow\mathbb{Z}_{\st p}$, where $\mathbb{U}$ is the output range of $\ses(.)$. 
\item appends $s'_{\st 2}, C'$ and $r'_{\st 2}$ to the view and outputs the view. 
\end{enumerate}

Next, we explain why the views in the ideal and real models are indistinguishable. 
Since values $g$ and $p$ have been picked uniformly at random in both models, they have identical distributions in the real and ideal models. 
Since $r_{\st 2}$ and $r'_{\st 2}$ have been picked uniformly at random from  $\mathbb{Z}_{\st p}$, they have identical distributions. Also, due to the security of $\ses(.)$ each share $s_{\st 2}$ is indistinguishable from a random value $s'_{\st 2}$, where $s'_{\st 2}\in \mathbb{U}$. Also, both $C$ and $C'$ have been picked uniformly at random from $\mathbb{Z}_{\st p}$; therefore, they have identical distribution. 
Thus, the two views are computationally indistinguishable, i.e., Relation \ref{equ::server-sim-} w.r.t. $\p_{\st 2}$ (in Section \ref{sec::sec-def}) holds.

\subsubsection{Corrupt Server $\p_{\st 1}$.}

In the real execution, $\p_{\st 1}$'s view is:  
$\view_{\st\p_1}^{\st DQ\text{-}OT}\big((m_{\st 0}, m_{\st 1}), \empt,$ $ \empt, s\big)=\{ g, C, p, s_{\st 1}, r_{\st 1}, \delta_{\st 0}, \delta_{\st 1}\}$. Ideal-model $\simm_{\st\p_{1}}$ works as follows.

\begin{enumerate}
\item initiates an empty view. It chooses a random generator $g$ and a large random prime number $p$. 
\item picks two random values $\delta'_{\st 0}, \delta'_{\st 1}\stackrel{\st\$}\leftarrow\mathbb{Z}_{\st p}$. 
\item picks two uniformly random values $s'_{\st 1}\stackrel{\st\$}\leftarrow\mathbb U$ and \\  $C', r'_{\st 1}\stackrel{\st\$}\leftarrow\mathbb{Z}_{\st p}$, where $\mathbb{U}$ is the output range of $\ses(.)$. 
\item appends  $s'_{\st 1}, C', r'_{\st 1}, \delta'_{\st 0}, \delta'_{\st 1}$ to the view and outputs the view. 
\end{enumerate}

Now, we explain why the views in the ideal and real models are indistinguishable. Values $g$ and $p$ have been picked uniformly at random in both models. Hence, $g$ and $p$ in the real and ideal models have identical distributions (pair-wise).

Recall, in the real model, $\p_{\st 1}$ receives $\delta_{\st s_{\st 2}}=g^{\st r_{\st 2}}$ and  $\delta_{\st 1- s_2}=g^{\st a- r_2}$ from  $\p_{\st 2}$. 
Since $a$ and $r_{\st 2}$ have been picked uniformly at random and unknown to the adversary due to DL assumption,   $\delta_{\st s_2}$ and $\delta_{\st 1-{s_2}}$ (or $\delta_{\st 0}$ and $\delta_{\st 1}$) have identical distributions and are indistinguishable from random values (of the same field). 

In the ideal model, $\delta'_{\st 0}$ and  $\delta'_{\st 1}$ have been picked uniformly at random; therefore, they have identical distributions too. Moreover,  $\delta_{\st s}, \delta_{\st 1-s}, \delta'_{\st 0}$, and $\delta'_{\st 1}$  have been defined over the same field, $\mathbb{Z}_{\st p}$. So, they have identical distributions and are indistinguishable. 
Due to the security of $\ses(.)$ each share $s_{\st 1}$ is indistinguishable from a random value $s'_{\st 1}$, where $s'_{\st 1}\in \mathbb{U}$. Also, $(r_{\st 1}, C)$ and $(r'_{\st 1}, C')$ have identical distributions, as they are picked uniformly at random from $\mathbb{Z}_{\st p}$.

Hence, the two views are computationally indistinguishable, i.e., Relation \ref{equ::server-sim-} w.r.t. $\p_{\st 1}$ (in Section \ref{sec::sec-def}) holds.
\hfill\(\Box\) 
\end{proof}


\subsection{Proof of Correctness}\label{sec::DQ-OT-proof-of-correctness}

In this section, we discuss why the correctness of DQ-OT always holds. Recall, in the original OT of Naor and Pinkas \cite{Efficient-OT-Naor}, the random value $a$ (i.e., the discrete logarithm of random value $C$) is inserted by receiver \re into the query $\beta_{\st 1-s}$ whose index (i.e., $1-s$) is not interesting to \re while the other query $\beta_{\st s}$ is free from value $a$. As we will explain below, in our DQ-OT, the same applies to the final queries that are sent to \se.  
Briefly, in DQ-OT, when:

\begin{itemize}

\item[$\bullet$]  $s=s_{\st 1}\oplus s_{\st 2}=1$ (i.e., when $s_{\st 1}\neq s_{\st 2}$), then $a$ will always appear in $\beta_{\st 1-s}=\beta_{\st 0}$; however, $a$ will not appear in $\beta_{\st 1}$. 

\item[$\bullet$] $s=s_{\st 1}\oplus s_{\st 2} = 0$ (i.e., when $s_{\st 1}=s_{\st 2}$), then $a$ will always appear in $\beta_{\st 1-s}=\beta_{\st 1}$; but $a$ will not appear in $\beta_{\st 0}$.

\end{itemize}

This is evident in Table~\ref{tab:analysis} which shows what
$\delta_{\st i}$ and $\beta_{\st j}$ are for the different values of $s_{\st 1}$ and $s_{\st 2}$. Therefore, the query pair ($\beta_{\st 0}, \beta_{\st 1}$) has the same structure as it has in \cite{Efficient-OT-Naor}. 


\begin{table}[h]
\begin{center}
\scalebox{1}{
  \begin{tabular}{|l||p{4.5cm}|p{4.5cm}|}
\hline
    &\multicolumn{1}{c|}{$s_{\st 2}=0$}&\multicolumn{1}{c|}{$s_{\st 2}=1$}\\
\hline
\hline
&\cellcolor{gray!20}$\delta_{\st 0} = g^{\st r_2}$, \hspace{3mm} $\delta_{\st 1} = g^{\st a-r_2}$&\cellcolor{gray!20} $\delta_{\st 0}=g^{\st a-r_2}$,\hspace{3mm} $\delta_{\st 1}=g^{\st r_2}$\\
\multirow{-2}{*}{$s_{\st 1}=0$}&\cellcolor{gray!20}$\beta_{\st 0}=g^{\st r_2+r_1}$, \hspace{3mm} $\beta_{\st 1}=g^{\st a-r_2-r_1}$&\cellcolor{gray!20}$\beta_{\st 0}=g^{\st a-r_2+r_1}$, \hspace{3mm} $\beta_{\st 1}=g^{\st r_2-r_1}$\\
\hline
&\cellcolor{gray!20}$\delta_{\st 0}=g^{\st r_2}$, \hspace{3mm} $\delta_{\st 1}=g^{\st a-r_2}$&\cellcolor{gray!20}$\delta_{\st 0}=g^{\st a-r_2}$, \hspace{3mm} $\delta_{\st 1}=g^{\st r_2}$\\
\multirow{-2}{*}{$s_{\st 1}=1$}&\cellcolor{gray!20}$\beta_{\st 0}=g^{\st a-r_2-r_1}$, \hspace{3mm} $\beta_{\st 1}=g^{\st r_2+r_1}$&\cellcolor{gray!20}$\beta_{\st 0}=g^{\st r_2-r_1}$, \hspace{3mm} $\beta_{\st 1}=g^{\st a-r_2+r_1}$\\
\hline
\end{tabular}
}
\end{center}
\caption{$\delta_{\st i}$ and $\beta_{\st j}$ are for the different values of $s_{\st 1}$ and $s_{\st 2}$.
We express each value as a power of $g$.}
\label{tab:analysis}
\end{table}

Next, we show why, in DQ-OT, \re can extract the correct message, i.e., $m_s$. Given \se's reply pair $(e_{\st 0}, e_{\st 1})$ and its original index $s$, \re knows which element to pick from the response pair, i.e., it picks $e_{\st s}$.  
%
Moreover, given $g^{\st y_s}\in e_{\st s}$, \re can  recompute $\h(g^{\st y_s})^{\st x}$, as it knows the value of $s, s_{\st 1}$, and $s_{\st 2}$.  Specifically, as Table~\ref{tab:analysis} indicates, when: 

\begin{itemize}
\item[$\bullet$]  $\overbrace{(s= s_{\st 1}=s_{\st 2}=0)}^{\st \text{Case 
 1}}$ or $\overbrace{(s=s_{\st 1}=1 \wedge s_{\st 2}=0)}^{\st \text{Case 
 2}}$, then \re can set $x = r_{\st 2} + r_{\st 1}$.

\begin{itemize}

 \item In Case 1, it holds $\h((g^{\st y_0})^{\st x})= \h((g^{\st y_0})^{\st r_2 +r_1})=q$. Also, $e_{\st 0}=\h(\beta_{\st 0}^{\st y_0})\oplus m_{\st 0}=\h((g^{\st r_2+r_1})^{\st y_0})\oplus m_{\st 0}$. Thus, $q\oplus e_{\st 0}=m_{\st 0}$.
 
 \item In Case 2, it holds $\h((g^{\st y_1})^{\st x})= \h((g^{\st y_1})^{\st r_2 +r_1})=q$. Moreover, $e_{\st 1}=\h(\beta_{\st 1}^{\st y_1})\oplus m_{\st 1}=\h((g^{\st r_2+r_1})^{\st y_1})\oplus m_{\st 1}$. Hence, $q\oplus e_{\st 1}=m_{\st 1}$.

 \end{itemize}

 \vspace{.6mm} 

\item[$\bullet$]  $\overbrace{(s=0 \wedge s_{\st 1}=s_{\st 2}=1)}^{\st \text{Case 
 3}}$ or $\overbrace{(s=s_{\st 2}=1 \wedge s_{\st 1}=0)}^{\st \text{Case 
 4}}$, then \re can set $x = r_{\st 2} - r_{\st 1}$.

\begin{itemize}

 \item In Case 3, it holds $\h((g^{\st y_0})^{\st x})= \h((g^{\st y_0})^{\st r_2 - r_1})=q$. On the other hand, $e_{\st 0}=\h(\beta_{\st 0}^{\st y_0})\oplus m_{\st 0}=\h((g^{\st r_2 - r_1})^{\st y_0})\oplus m_{\st 0}$. Therefore, $q\oplus e_{\st 0}=m_{\st 0}$.
 
 \item In Case 4, it holds $\h((g^{\st y_1})^{\st x})= \h((g^{\st y_1})^{\st r_2 -r_1})=q$. Also, $e_{\st 1}=\h(\beta_{\st 1}^{\st y_1})\oplus m_{\st 1}=\h((g^{\st r_2-r_1})^{\st y_1})\oplus m_{\st 1}$. Hence, $q\oplus e_{\st 1}=m_{\st 1}$.
 
 \end{itemize}

\end{itemize}

We conclude that DQ-OT always allows honest \re to recover the message of its interest, i.e., $m_s$.


\section{Delegated-Unknown-Query OT}\label{sec::DUQ-OT}

%

In certain cases, the receiver itself may not know the value of query $s$. Instead, the query is issued by a third-party query issuer (\tp).  
In this section, we present a new variant of \dqot, called Delegated-Unknown-Query 1-out-of-2 OT (\duqot). It enables \tp to issue the query while (a) preserving the security of  \dqot and (b) preserving the privacy of query $s$ from \re.

%
%
%
%
\subsection{Security Definition}\label{sec::DUQ-OT-definition}

The functionality that \duqot computes takes as input (a) a pair of messages $(m_{\st 0}, m_{\st 1})$ from \se, (b) empty strings \empt from $\p_{\st 1}$, (c)  \empt from $\p_{\st 2}$, (d)   \empt from $\re$, and (e) the index $s$ (where $s\in \{0, 1\}$) from \tp. It outputs an empty string $\empt$ to  \se, \tp, $\p_{\st 1}$, and $\p_{\st 2}$, and outputs the message with index $s$, i.e., $m_{\st s}$, to \re. More formally, we define the functionality as: 
 $\mathcal{F}_{\scriptscriptstyle\duqot}:\big((m_{\st 0}, m_{\st 1}), \empt, \empt, \empt, s\big) \rightarrow (\empt, \empt, \empt, m_{\st s}, \empt)$. 
Next, we present a formal definition of \duqot.

\begin{definition}[\duqot]\label{def::DUQ-OT-sec-def} Let $\mathcal{F}_{\scriptscriptstyle\duqot}$ be the functionality defined above. We assert that protocol $\Gamma$ realizes $\mathcal{F}_{\scriptscriptstyle\duqot}$ in the presence of passive adversary \se, \re, $\p_{\st 1}$, or $\p_{\st 2}$, if for  every PPT adversary \adv
in the real model, there exists a non-uniform PPT simulator \simm  in
the ideal model, such that:

\begin{equation}\label{equ::DUQ-OT-sender-sim-}
\begin{split}
\Big\{\simm_{\st\se}\big((m_{\st 0}, m_{\st 1}), \empt\big)\Big\}_{\st m_{\st 0}, m_{\st 1}, s}\stackrel{c}{\equiv}  \Big\{\view_{\st \se}^{\st \Gamma}\big((m_{\st 0}, m_{\st 1}), \empt,  \empt, \empt, s\big) \Big\}_{\st m_{\st 0}, m_{\st 1}, s}
\end{split}
\end{equation}
\begin{equation}\label{equ::DUQ-OT-server-sim-}
\begin{split}
\Big\{\simm_{\st\p_i}(\empt, \empt)\Big\}_{\st m_{\st 0}, m_{\st 1}, s}\stackrel{c}{\equiv} 
\Big\{\view_{\st\p_i}^{\st \Gamma}\big((m_{\st 0}, m_{\st 1}), \empt,  \empt, \empt, s\big) \Big\}_{\st m_0, m_1, s}
\end{split}
\end{equation}
\begin{equation}\label{equ::DUQ-OT-t-sim-}
\begin{split}
\Big\{\simm_{\st\tp}(s, \empt)\Big\}_{\st m_0, m_1, s}\stackrel{c}{\equiv}  \Big\{\view_{\st\tp}^{\st \Gamma}\big((m_{\st 0}, m_{\st 1}), \empt,  \empt, \empt, s\big) \Big\}_{\st m_0, m_1, s}
\end{split}
\end{equation}
\begin{equation}\label{equ::DUQ-OT-reciever-sim-}
\begin{split}
\Big\{\simm_{\st\re}\Big(\empt, \mathcal{F}_{\scriptscriptstyle\duqot}\big((m_{\st 0}, m_{\st 1}), \empt,  \empt,\empt, s\big)\Big)\Big\}_{\st m_0, m_1, s}\stackrel{c}{\equiv} \Big\{\view_{\st\re}^{\st \Gamma}\big((m_{\st 0}, m_{\st 1}), \empt,  \empt, \empt, s\big) \Big\}_{\st m_0, m_1, s}
\end{split}
\end{equation}

for all $i$,  $i\in \{1,2\}$. Since \duqot is a variant of \dqot, it also supports efficiency and \spc, as discussed in Section \ref{sec::sec-def}.

\end{definition}


\subsection{Protocol}\label{sec::DUQ-OT-protocol}
In this section, we present DUQ-OT that realizes \duqot.  

\subsubsection{Main Challenge to Overcome.} 
One of the primary differences between DUQ-OT and previous OTs in the literature (and DQ-OT) is that in DUQ-OT, \re does not know the secret index $s$. The knowledge of $s$ would help \re pick the suitable element from \se's response; for instance, in the DQ-OT, it picks  $e_s$ from $(e_{\st 0}, e_{\st 1})$. Then, it can extract the message from the chosen element. In DUQ-OT, to enable \re to extract the desirable message from \se's response without the knowledge of $s$, we rely on the following observation and technique. 
We know that (in any OT) after decrypting $e_{\st s-1}$, \re would obtain a value indistinguishable from a random value (otherwise, it would learn extra information about $m_{\st s-1}$). 
Therefore, if \se imposes a certain publicly known structure to messages $(m_{\st 0}, m_{\st 1})$, then after decrypting \se's response, only $m_{\st s}$ would preserve the same structure. 
In DUQ-OT, \se imposes a publicly known structure to $(m_{\st 0}, m_{\st 1})$  and then computes the response. Given the response, \re tries to decrypt \emph{every} message it received from \se and accepts only the result that has the structure.

\subsubsection{An Overview.} 
DUQ-OT operates as follows. First, \re picks two random values and sends each to a $\p_{\st i}$. Also, \tp splits the secret index $s$ into two shares and sends each share to a $\p_{\st i}$. Moreover, \tp selects a random value $r_{\st 3}$ and sends it to \re and \se. Given the messages receives from \re and \tp, each $\p_{\st i}$ generates queries the same way they do in DQ-OT.  
Given the final query pair and $r_{\st 3}$,  \se first appends $r_{\st 3}$ to $m_{\st 0}$ and $m_{\st 1}$ and then computes the response the same way it does in DQ-OT, with the difference that it also randomly permutes the elements of the response pair.  Given the response pair and $r_{\st 3}$, \re decrypts each element in the pair and accepts the result that contains $r_{\st 3}$. 
Figure \ref{fig::DQ-OT-with-unknown-query} presents DUQ-OT in more detail. 

\begin{figure}[!htbp]
\setlength{\fboxsep}{1pt}
\begin{center}
    \begin{tcolorbox}[enhanced, 
    drop fuzzy shadow southwest,
    colframe=black,colback=white]
\begin{enumerate}

\item \underline{\textit{$\se$-side Initialization:}} 
$\mathtt{Init}(1^{\st \lambda})\rightarrow pk$
\begin{enumerate}

\item chooses a sufficiently large prime number $p$.

\item selects random element
$C \stackrel{\st \$}\leftarrow \mathbb{Z}_p$ and generator $g$.
\item publishes $pk=(C, p, g)$. 

\end{enumerate}

\item \underline{\textit{\re-side Delegation:}} $\mathtt{\re.Request}( pk)\hspace{-.6mm}\rightarrow\hspace{-.6mm} req=(req_{\st 1}, req_{\st 2})$
\begin{enumerate}

\item picks two uniformly random values: 
$r_{\st 1}, r_{\st 2} \stackrel{\st\$}\leftarrow\mathbb{Z}_{\st p}$.

\item sends $req_{\st 1}= r_{\st 1}$ to $\p_{\st 1}$ and $req_{\st 2}=r_{\st 2}$ to $\p_{\st 2}$.

\end{enumerate}

\item \underline{\textit{$\tp$-side Query Generation:}}
$\mathtt{\tp.Request}(1^{\st \lambda}, s, pk)\hspace{-.7mm}\rightarrow\hspace{-.7mm} (req'_{\st 1}, req'_{\st 2}, sp_{\st \se})$

\begin{enumerate}

\item splits  the private index $s$ into two shares $(s_{\st 1}, s_{\st 2})$ by calling  $\ses(1^{\st\lambda}, s, 2, 2)\rightarrow (s_{\st 1}, s_{\st 2})$.

\item picks a uniformly random value: $r_{\st 3} \stackrel{\st\$}\leftarrow\{0,1\}^{\st\lambda}$.

\item sends $req'_{\st 1} =s_{\st 1}$ to  $\p_{\st 1}$, $req'_{\st 2}=s_{\st 2}$ to  $\p_{\st 2}$. It also sends secret parameter $sp_{\st \se}=r_{\st 3}$ to \se and $sp_{\st \re} = (req'_{\st 2}, sp_{\st \se})$ to \re.

\end{enumerate}

\item \underline{\textit{$\p_{\st 2}$-side Query Generation:}}
$\mathtt{\p_{\st 2}.GenQuery}(req_{\st 2}, req'_{\st 2}, pk)\rightarrow q_{\st 2}$

\begin{enumerate}

\item computes a pair of partial queries: 

 $$\delta_{\st s_2}= g^{\st r_2},\ \  \delta_{\st 1-s_2} = \frac{C}{g^{\st r_2}}$$
  
\item sends $q_{\st 2}=(\delta_{\st 0}, \delta_{\st 1})$ to  $\p_{\st 1}$. 

\end{enumerate}

\item\underline{\textit{$\p_{\st 1}$-side Query Generation:}}
$\mathtt{\p_{\st 1}.GenQuery}(req_{\st 1},req'_{\st 1}, q_{\st 2},pk)\rightarrow q_{\st 1}$

\begin{enumerate}

\item computes a pair of final queries as:  
%

$$\beta_{\st s_1}=\delta_{\st 0}\cdot g^{\st r_1},\ \ \beta_{\st 1-s_1}=\frac{\delta_{\st 1}} {g^{\st r_1}}$$

\item sends $q_{\st 1} =(\beta_{\st 0}, \beta_{\st 1})$ to  $\se$.

\end{enumerate}

\item\underline{\textit{\se-side Response Generation:}} 
$\mathtt{GenRes}(m_{\st 0}, m_{\st 1}, pk, q_{\st 1}, sp_{\st \se})\rightarrow res$

\begin{enumerate}

\item aborts if  $C \neq \beta_{\st 0}\cdot \beta_{\st 1}$.
\item picks two uniformly random values: 
$y_{\st 0}, y_{\st 1}  \stackrel{\st\$}\leftarrow\mathbb{Z}_{\st p}$.
    
\item computes a response pair $(e_{\st 0}, e_{\st 1})$ as follows:
%

$$e_{\st 0} := (e_{\st 0,0}, e_{\st 0,1}) = (g^{\st y_0}, \g(\beta_{\st 0}^{\st y_0}) \oplus (m_{\st 0}||r_{\st 3}))$$
$$e_{\st 1} := (e_{\st 1,0}, e_{\st 1,1}) = (g^{\st y_1}, \g(\beta_{\st 1}^{\st y_1}) \oplus (m_{\st 1}||r_{\st 3}))$$

\item randomly permutes the elements of the pair $(e_{\st 0}, e_{\st 1})$ as follows: $\pi(e_{\st 0}, e_{\st 1})\rightarrow ({e}'_{\st 0}, {e}'_{\st 1})$.

\item sends $res=(e'_{\st 0}, e'_{\st 1})$ to \re. 

\end{enumerate}

\item\underline{\textit{\re-side Message Extraction:}} 
$\mathtt{Retreive}(res, req, pk, sp_{\st \re})\hspace{-.6mm}\rightarrow\hspace{-.8mm} m_{\st s}$

\begin{enumerate}

\item sets $x=r_{\st 2}+r_{\st 1}\cdot(-1)^{\st s_2}$. 

\item retrieves message $m_{\st s}$ as follows. 
 $\forall i, 0\leq i\leq1:$

\begin{enumerate}

\item sets $y=\g(({e}'_{\st i, 0})^{\st x})\oplus {e}'_{\st i, 1}$.

\item  calls $\parse(\gamma, y)\rightarrow (\ux, \uy)$.

\item sets  $m_{\st s}=\ux$, if $\uy=r_{\st 3}$.

\end{enumerate}

\end{enumerate}

\end{enumerate}
\end{tcolorbox}
\end{center}
    \caption{DUQ-OT: Our $1$-out-of-$2$ OT that supports query delegation while preserving the privacy of query from \re. 
    }
    \label{fig::DQ-OT-with-unknown-query}
\end{figure}

\begin{theorem}\label{theo::DUQ-OT-sec}
Let $\mathcal{F}_{\scriptscriptstyle\duqot}$ be the functionality defined in Section \ref{sec::DUQ-OT-definition}. If  
DL, CDH, and RO assumptions hold and random permutation $\pi(.)$  is secure, then DUQ-OT (presented in Figure \ref{fig::DQ-OT-with-unknown-query}) securely computes $\mathcal{F}_{\scriptscriptstyle\duqot}$ in the presence of semi-honest adversaries, 
%
%
w.r.t. Definition \ref{def::DUQ-OT-sec-def}. 

\end{theorem}


\subsection{DUQ-OT's Security Proof}\label{sec::DUQ-OT-Security-Proof}
Below, we prove  DUQ-OT's security theorem, i.e., Theorem \ref{theo::DUQ-OT-sec}.  Even though the proofs of DUQ-OT and DQ-OT have similarities, they have significant differences too. Thus, for the sake of completeness, we present a complete proof for DUQ-OT.

\begin{proof}
We consider the case where each party is corrupt, at a time.

\subsubsection{Corrupt Receiver \re.} In the real execution, \re's view is: 
$\view_{\re}^{\st DUQ\text{-}OT}\big(m_{\st 0}, m_{\st 1}, \empt,  \empt, \empt, s\big) = \{r_{\st\re}, g, C, p,  r_{\st 3}, s_{\st 2}, $ $e'_{\st 0}, e'_{\st 1}, m_{\st s} \}$, 
where $r_{\st \re}$ is the outcome of the internal random coin of  \re and is used to generate $(r_{\st 1}, r_{\st 2})$.  
 Below, we construct an idea-model simulator  $\simm_{\st\re}$ which receives $m_{\st s}$ from \re. 
 
 \begin{enumerate}
 \item initiates an empty view and appends uniformly random coin $r'_{\st\re}$ to it, where $r'_{\st\re}$ will be used to generate \re-side randomness, i.e., $(r'_{\st 1}, r'_{\st 2})$. 

 \item selects a random generator $g$ and a large random prime number $p$.
 \item sets response $(\bar{e}'_{\st 0}, \bar{e}'_{\st 1})$ as follows: 
 \begin{itemize}
 %
 %
 \item picks random values: $C', r'_{\st 1}, r'_{\st 2}, y'_{\st 0}, y'_{\st 1}\stackrel{\st\$}\leftarrow\mathbb{Z}_{\st p}$, $ r'_{\st 3}\stackrel{\st\$}\leftarrow\{0, 1\}^{\st\lambda}$, $s' \stackrel{\st\$}\leftarrow\{0,1\}$, and $u\stackrel{\st\$}\leftarrow\{0, 1\}^{\st\sigma+\lambda}$. 
 
 \item sets $x=r'_{\st 2}+r'_{\st1}\cdot (-1)^{\st s'}$ and $\beta'_{\st 0}=g^{\st x}$. 
 \item sets $\bar{e}_{\st 0} = (g^{\st y'_0}, \g(\beta'^{\st y'_0}_{\st 0})\oplus (m_{\st s}|| r'_{\st 3}))$ and $\bar{e}_{\st 1} = (g^{\st y'_{1}}, u)$. 

 \item randomly permutes the element of  pair $(\bar{e}_{\st 0}, \bar{e}_{\st 1})$. Let $(\bar{e}'_{\st 0}, \bar{e}'_{\st 1})$ be the result. 

 \end{itemize}
 \item  appends $(g, C', p, r'_{\st 3}, s', \bar{e}'_{\st 0}, \bar{e}'_{\st 1}, m_{\st s})$ to the view and outputs the view.
 \end{enumerate}
 
 Next, we argue that the views in the ideal and real models are indistinguishable.  
 As we are in the semi-honest model, the adversary picks its randomness according to the protocol description; therefore, $r_{\st \re}$ and $r'_{\st \re}$ model have identical distributions, the same holds for values $(r_{\st 3}, s_{\st 2})$ in the real model and $(r'_{\st 3}, s')$ in the ideal model, component-wise.  
Furthermore, because values $g$ and $p$ have been selected uniformly at random in both models, they  have identical distributions in the real and ideal models.

 For the sake of simplicity, in the ideal mode let $\bar{e}'_{j}=\bar{e}_{\st 1} = (g^{\st y'_{\st 1}}, u)$ and in the real model let $e'_{\st i}=e_{\st 1-s}=(g^{\st y_{1-s}}, \g(\beta^{\st y_{1-s}}_{\st 1-s})\oplus (m_{\st 1-s}||r_{\st 3}))$, where $i, j\in\{0,1 \}$.  We will explain that  $e'_{\st i}$ in the real model and $\bar{e}'_{\st j}$ in the ideal model are indistinguishable. 
 
 In the real model, it holds that $e_{\st 1-s}=(g^{\st y_{1-s}}, \g(\beta^{\st y_{1-s}}_{\st 1-s})\oplus (m_{1-s}||r_{\st 3}))$, where $\beta^{\st y_{1-s}}_{\st 1-s}=\frac{C}{g^{\st x}}=g^{\st a-x}$. Since $y_{\st 1-s}$ in the real model and $y'_{\st 1}$ in the ideal model have been picked uniformly at random and unknown to the adversary, $g^{\st y_{1-s}}$ and $g^{\st y'_{1}}$  have identical distributions.

 Moreover, in the real model, given $C=g^{\st a}$, because of DL problem, $a$ cannot be computed by a PPT adversary. 
 Furthermore, due to CDH assumption, \re cannot compute  $\beta^{\st y_{1-s}}_{\st 1-s}$ (i.e., the input of $\g(.)$), given $g^{\st y_{\st 1-s}}$ and $g^{\st a-x}$. We know that $\g(.)$ is considered as a random oracle and its output is indistinguishable from a random value. Therefore,  $\g(\beta^{\st y_{\st 1-s}}_{\st 1-s})\oplus (m_{\st 1-s} || r_{\st 3})$ in the real model and $u$ in the ideal model are indistinguishable. This means that $e_{\st 1-s}$ and $\bar{e}'_{\st j}$ are indistinguishable too, due to DL, CDH, and RO assumptions.

 Moreover, since (i) $y_{\st s}$ in the real model and $y'_{\st 0}$ in the ideal model have picked uniformly at random and (ii) the decryption of both $e'_{\st 1-i}$ and $\bar{e}'_{\st 1-j}$ contain $m_{\st s}$, $e'_{\st 1-i}$ and $\bar{e}'_{\st 1-j}$ have identical distributions. $m_{\st s}$ also has identical distribution in both models.  Both $C$ and $C'$ have also been picked uniformly at random from $\mathbb{Z}_{\st p}$; therefore, they have identical distributions. 
 
 In the ideal model, $\bar{e}_{\st 0}$ always contains encryption of actual message $m_{\st s}$ while $\bar{e}_1$ always contains a dummy value $u$. However, in the ideal model the elements of pair  $(\bar{e}_{\st 0}, \bar{e}_{\st 1})$ and in the real model the elements of  pair $(e_{\st 0}, e_{\st 1})$ have been randomly permuted, which results in $(\bar{e}'_{\st 0}, \bar{e}'_{\st 1})$ and $(e'_{\st 0}, e'_{\st 1})$ respectively. Therefore, the permuted pairs have identical distributions too. 

We conclude that the two views are computationally indistinguishable, i.e., Relation \ref{equ::DUQ-OT-reciever-sim-} (in Section \ref{sec::DUQ-OT-definition}) holds.

\subsubsection{Corrupt Sender \se.} In the real model, \se's view is: 
$\view_{\se}^{\st DUQ\text{-}OT}\big((m_{\st 0}, m_{\st 1}), \empt,  \empt ,\empt, s\big)= \{r_{\st\se}, C, r_{\st 3}, \beta_{\st 0}, \beta_{\st 1}\}$, where $r_{\st \se}$ is the outcome of the internal random coin of  \se.  Next, we construct an idea-model simulator  $\simm_{\se}$ which receives $\{m_{\st 0}, m_{\st 1}\}$ from \se. 

\begin{enumerate}
\item initiates an empty view and appends uniformly random coin $r'_{\st\se}$ to it, where $r'_{\st\se}$ will be used to generate random values for \se.
\item picks  random values $C', r'\stackrel{\st\$}\leftarrow\mathbb{Z}_{\st p}, r'_{\st 3}\stackrel{\st \$}\leftarrow\{0, 1\}^{\st \lambda}$.

\item sets $\beta'_{\st 0}=g^{\st r'}$ and $\beta'_{1}=\frac{C'}{g^{\st r'}}$.

\item appends $C', r'_{\st 3}, \beta'_{\st 0}$, and $\beta'_{\st 1}$ to the view and outputs the view. 

\end{enumerate}

Next, we explain why the two views in the ideal and real models are indistinguishable. Recall, in the real model, $(\beta_{\st s}, \beta_{\st 1-s})$ have the following form: $\beta_{s}=g^{\st x}$ and  $\beta_{\st 1-s}=g^{\st a-x}$, where $a=DL(C)$ and $C=g^{\st a}$. 

In this ideal model, as $a$ and $x$ have been picked uniformly at random and unknown to the adversary, due to DL assumption,   $\beta_{\st s}$ and $\beta_{\st 1-s}$ have identical distributions and are indistinguishable. 

In the ideal model, $r'$ and $C'$ have been picked uniformly at random and we know that $a'$ in $C'=g^{\st a'}$ is a uniformly random value, unknown to the adversary; thus, due to DL assumption, $\beta'_{\st 0}$ and $\beta'_{\st 1}$ have identical distributions too. The same holds for values $C$ and $C'$. 
Moreover, values $\beta_{\st s}, \beta_{\st 1-s}, \beta'_{\st 0}$, and $\beta'_{\st 1}$  have been defined over the same field, $\mathbb{Z}_{\st p}$. Thus, they have identical distributions and are indistinguishable.  The same holds for values $r_{\st 3}$ in the real model and $r'_{\st 3}$ in the ideal model. 

Therefore, the two views are computationally indistinguishable, i.e., Relation \ref{equ::DUQ-OT-sender-sim-} (in Section \ref{sec::DUQ-OT-definition}) holds.

\subsubsection{Corrupt Server $\p_{\st 2}$.}  
 
 In the real execution, $\p_{\st 2}$'s view is: 
$\view_{\st \p_2}^{\st DUQ\text{-}OT}\big((m_{\st 0}, m_{\st 1}), \empt, \empt, \empt, s\big)=\{ g, C, p, s_{\st 2}, r_{\st 2}\}$. 
Below, we show how an ideal-model simulator $\simm_{\st \p_{\st 2}}$ works. 

\begin{enumerate}
\item initiates an empty view. It chooses a random generator $g$ and a large random prime number $p$.

\item picks two uniformly random values $s'_{\st 2}\stackrel{\st\$}\leftarrow\mathbb U$ and $C', r'_{\st 2}\stackrel{\st\$}\leftarrow\mathbb{Z}_{\st p}$, where $\mathbb{U}$ is the output range of $\ses(.)$. 
\item appends $s'_{\st 2}, C'$ and $r'_{\st 2}$ to the view and outputs the view. 
\end{enumerate}

Next, we explain why the views in the ideal and real models are indistinguishable. 
 Values $g$ and $p$ have been selected uniformly at random in both models. Thus, they  have identical distributions in the real and ideal models. 
Since $r_{\st 2}$ and $r'_{\st 2}$ have been picked uniformly at random from  $\mathbb{Z}_{\st p-1}$, they have identical distributions. 

Also, due to the security of $\ses(.)$ each share $s_{\st 2}$ is indistinguishable from a random value $s'_{\st 2}$, where $s'_{\st 2}\in \mathbb{U}$. Also, both $C$ and $C'$ have been picked uniformly at random from $\mathbb{Z}_{\st p}$. Therefore, they have identical distributions.

Thus, the two views are computationally indistinguishable, i.e., Relation \ref{equ::DUQ-OT-server-sim-} w.r.t. $\p_{\st 2}$ (in Section \ref{sec::DUQ-OT-definition}) holds.

\subsubsection{Corrupt Server $\p_{\st 1}$.}

In the real execution, $\p_{\st 1}$'s view is: 
$\view_{\st\p_1}^{\st DUQ\text{-}OT}\big((m_{\st 0}, m_{\st 1}), \empt, \empt, \empt, s\big)=\{ g, C, p, s_{\st1}, r_{\st 1},$ $ \delta_{\st 0}, \delta_{\st 1}\}$. Ideal-model $\simm_{\st \p_{1}}$ works as follows.

\begin{enumerate}
\item initiates an empty view. It chooses a large random prime number $p$ and a random generator $g$.
\item picks two random values $\delta'_{\st 0}, \delta'_{\st 1}\stackrel{\st\$}\leftarrow\mathbb{Z}_{\st p}$. 
\item picks two uniformly random values $s'_{\st 1}\stackrel{\st\$}\leftarrow\mathbb U$ and $C', r'_{\st 1}\stackrel{\st\$}\leftarrow\mathbb{Z}_{\st p}$, where $\mathbb{U}$ is the output range of $\ses(.)$. 
\item appends  $s'_{\st 1}, g, C', p, r'_{\st 1}, \delta'_{\st 0}, \delta'_{\st 1}$ to the view and outputs the view. 
\end{enumerate}

Now, we explain why the views in the ideal and real models are indistinguishable. Recall, in the real model, $\p_{\st 1}$ receives $\delta_{\st s_2}=g^{\st r_2}$ and  $\delta_{\st 1- s_2}=g^{\st a- r_2}$ from  $\p_{\st 2}$. 
Since $a$ and $r_{\st 2}$ have been picked uniformly at random and unknown to the adversary due to DL assumption,   $\delta_{\st s_2}$ and $\delta_{\st 1-{\st s_2}}$ (or $\delta_{\st 0}$ and $\delta_{\st 1}$) have identical distributions and are indistinguishable from random values (of the same field). 

In the ideal model, $\delta'_{\st 0}$ and  $\delta'_{\st 1}$ have been picked uniformly at random; therefore, they have identical distributions too. Moreover, values $\delta_{\st s}, \delta_{\st 1-s}, \delta'_{\st 0}$, and $\delta'_{\st 1}$  have been defined over the same field, $\mathbb{Z}_{\st p}$. So, they have identical distributions and are indistinguishable. 
Due to the security of $\ses(.)$ each share $s_{\st 1}$ is indistinguishable from a random value $s'_{1}$, where $s'_{\st 1}\in \mathbb{U}$. Furthermore, $(r_{\st 1}, C)$ and $(r'_{\st 1}, C')$ have identical distributions, as they are picked uniformly at random from $\mathbb{Z}_{\st p}$.  Values $g$ and $p$ in the real and ideal models have identical distributions as they have been picked uniformly at random.

Hence, the two views are computationally indistinguishable, i.e., Relation \ref{equ::DUQ-OT-server-sim-} w.r.t. $\p_{\st 2}$ (in Section \ref{sec::DUQ-OT-definition})  holds.

\subsubsection{Corrupt \tp.} T's view can be easily simulated.  It has input $s$, but it receives no messages from its counterparts and receives no output from the protocol. Thus, its real-world view is defined as 
 $\view_{\st\tp}^{\st DUQ\text{-}OT}\big((m_{\st 0}, m_{\st 1}), \empt, \empt, \empt,$ $ s\big)=\{r_{\st \tp}, g, C,  p\}$, 
  where  $r_{\st \tp}$ is the outcome of the internal random coin of  \tp and is used to generate random values.   
Ideal-model $\simm_{\st\tp}$ initiates an empty view, picks  $r'_{\st \tp}$, $g, C$, and $p$ uniformly at random, and adds them to the view. Since, in the real model, the adversary is passive, then it picks its randomness according to the protocol's description; thus, $r_{\st \tp}, g, C, p$ and $r'_{\st \tp}, g, C, p$ have identical distributions. 

Thus, the two views are computationally indistinguishable, i.e., Relation \ref{equ::DUQ-OT-t-sim-} (in Section \ref{sec::DUQ-OT-definition}) holds.
\hfill\(\Box\) 
\end{proof}

\section{Delegated-Query Multi-Receiver Oblivious Transfers}\label{sec::Multi-Receiver-OT}

In this section, we present two new variants of \dqot; namely, (1) Delegated-Query Multi-Receiver OT (\dqothf) and (2) Delegated-Unknown-Query Multi-Receiver OT (\duqothf). They are suitable for the \emph{multi-receiver} setting in which the sender maintains a (large) database containing $z$ pairs of messages $\bm m=[(m_{\st 0, 0},m_{\st 1, 0}),\ldots,$ $ (m_{\st 0, z-1},$ $m_{\st 1, z-1})]$. 

In this setting, each pair, say $v$-th pair $(m_{\st 0, v},$ $m_{\st 1, v})\in \bm m$ is related to  a receiver,  $\re_{\st j}$, where $0\leq v\leq z-1$. Both variants (in addition to offering the efficiency, \spc, and security guarantee of \dqot) ensure that (i) a receiver learns nothing about the total number of receivers/pairs (i.e., $z$) and (ii) the sender learns nothing about which receiver is sending the query, i.e., a message pair's index for which a query was generated.  In the remainder of this section, we discuss these new variants.

\subsection{Delegated-Query Multi-Receiver OT}\label{sec::DQ-OT-HF}
The first variant \dqothf considers the setting where server $\p_{\st 1}$ or $\p_{\st 2}$ knows a client's related pair's index in the sender's database.


 \subsubsection{Security Definition.}\label{sec::DQOT-HF}

The functionality that \dqothf computes takes as input (i)  a vector of messages $\bm{m}=[(m_{\st 0, 0},m_{\st 1, 0}),\ldots,$ $ (m_{\st 0, z-1},$ $m_{\st 1, z-1})]$ from \se, (ii) an index $v$ of a pair in $\bm{m}$ from $\p_{\st 1}$, (iii)  empty string  \empt from $\p_{\st 2}$, and (iv) the index $s$ (where $s\in \{0, 1\}$) from \re. It outputs an empty string $\empt$ to  \se, $z$ to $\p_{\st 1}$, $\empt$ to $\p_{\st 2}$, and outputs to \re $s$-th message from $v$-th pair in the vector, i.e., $m_{\st s, v}$. Formally, we define the functionality as: 
 $\mathcal{F}_{\scriptscriptstyle\dqothf}:\big([(m_{\st 0, 0},m_{\st 1, 0}),\ldots,$ $ (m_{\st 0, z-1},$ $m_{\st 1, z-1})], v, \empt, s\big) \rightarrow (\empt, z, \empt, m_{\st s, v})$, where $v\in\{0,\ldots, z-1\}$. Next, we present a formal definition of \dqothf.

\begin{definition}[\dqothf]\label{def::DQ-OT-HF-sec-def} Let $\mathcal{F}_{\scriptscriptstyle\dqothf}$ be the functionality defined above. We say that protocol $\Gamma$ realizes $\mathcal{F}_{\scriptscriptstyle\dqothf}$ in the presence of passive adversary \se, \re, $\p_{\st 1}$, or $\p_{\st 2}$, if for  every non-uniform PPT adversary \adv
in the real model, there exists a non-uniform PPT simulator \simm  in
the ideal model, such that:

\begin{equation}\label{equ::sender-sim-DQ-OT-HF}
\Big\{\simm_{\st\se}\big(\bm{m}, \empt\big)\Big\}_{\st \bm{m}, v, s}\stackrel{c}{\equiv} \Big\{\view_{\st\se}^{\st \Gamma}\big(\bm{m}, v,  \empt, s\big) \Big\}_{\st \bm{m}, v, s}
\end{equation}
\begin{equation}\label{equ::server1-sim-DQ-OT-HF}
\Big\{\simm_{\st\p_1}(v, z)\Big\}_{\st \bm{m}, v, s}\stackrel{c}{\equiv} \Big\{\view_{\st\p_1}^{\st \Gamma}\big(\bm{m}, v,  \empt, s\big) \Big\}_{\st \bm{m}, v, s}
\end{equation}
\begin{equation}\label{equ::server2-sim-DQ-OT-HF}
\Big\{\simm_{\st\p_2}(\empt, \empt)\Big\}_{\st \bm{m}, v, s}\stackrel{c}{\equiv} \Big\{\view_{\st\p_2}^{ \st\Gamma}\big(\bm{m}, v,  \empt, s\big) \Big\}_{\st \bm{m}, v, s}
\end{equation}
\begin{equation}\label{equ::reciever-sim-DQ-OT-HF}
\begin{split}
\Big\{\simm_{\st\re}\Big(s, \mathcal{F}_{\scriptscriptstyle\dqothf}\big(\bm{m}, v,  \empt, s\big)\Big)\Big\}_{\st \bm{m}, v,  s}\stackrel{c}{\equiv}    \Big\{\view_{\st\re}^{ \st\Gamma}\big(\bm{m}, v,  \empt, s\big) \Big\}_{\st \bm{m}, v, s}
\end{split}
\end{equation}
where $\bm{m}=[(m_{\st 0, 0},m_{\st 1, 0}),\ldots, (m_{\st 0, z-1}, m_{\st 1, z-1})]$.

\end{definition}


\subsubsection{Strawman Approaches.}\label{sec::OT-HFstrawman-aprroach}

One may consider using one of the following ideas in the multi-receiver setting: 

\begin{enumerate}[leftmargin=5mm]
    \item \underline{\textit{Using an existing single-receiver OT, e.g., in \cite{IshaiKNP03}}}, employing one of the following approaches:

    \begin{itemize}
        \item \textit{Approach 1}: receiver $\re_{\st j}$ sends a standard OT query to \se which computes the response for all $z$ pairs of messages. Subsequently, \se sends $z$ pair of responses to receiver $\re_{\st j}$ which discards all pairs from the response except for $v$-th pair. $\re_{\st j}$ extracts its message $m_{\st v}$ from the selected pair, similar to a regular $1$-out-of-$2$ OT.  However, this approach results in the leakage of the entire database size to $\re_{\st j}$.

        \item \textit{Approach 2}: $\re_{\st j}$ sends a standard OT query to \se, along with the index $v$ of its record.  This can be perceived as if \se holds a single record/pair. Accordingly, \se generates a response in the same manner as it does in regular $1$-out-of-$2$ OT. Nevertheless, Approach 2 leaks to $\se$ the index $v$ of the record that $\re_{\st j}$ is interested.
    \end{itemize}
    \item \underline{\textit{Using an existing multi-receiver OT, e.g., in \cite{CamenischDN09}}}. This will also come with a privacy cost. The existing multi-receiver OTs reveal the entire database's size to each receiver (as discussed in Section \ref{sec::multi-rec-ot}). In this scenario, a receiver can learn the number of private records other companies have in the same database. This type of leakage is particularly significant, especially when coupled with specific auxiliary information. 
\end{enumerate}
 Hence, a fully private multi-receiver OT is necessary to ensure user privacy in real-world cloud settings.

\subsubsection{Protocol.}\label{sec::rdqothf}
We present \rdqothf that realizes \dqothf. We build \rdqothf upon DQ-OT (presented in Figure \ref{fig::DQ-OT}). \rdqothf relies on our observation that in DQ-OT, given the response of \se, $\p_{\st 1}$ cannot learn anything, e.g., about the plaintext messages $m_{\st i}$ of  \se. Below, we formally state it.

\begin{lemma}\label{lemma::two-pairs-indis-}
Let $g$ be a generator of a group $\mathbb{G}$ (defined in Section \ref{sec::dh-assumption}) whose order is a prime number $p$ and $\log_{\st 2}(p)=\lambda$ is a security parameter. Also, let $(r_{\st 1}, r_{\st 2}, y_{\st 1}, y_{\st 2})$ be elements of $\mathbb{G}$ picked uniformly at random, $C=g^{\st a}$ be a random public value whose discrete logarithm is unknown, $(m_{\st 0}, m_{\st 1})$ be two arbitrary messages, and $\h$ be a hash function modelled as a RO (as defined in Section \ref{sec::notations}), where its output size is $\delta$-bit.  Let $\gamma=\stackrel{\st+}-r_{\st 1} \stackrel{\st+}-r_{\st 2}$, $\beta_{\st 0}=g^{\st a+\gamma}$, and $\beta_{\st 1}=g^{\st a-\gamma}$.  If DL, RO, and CDH assumptions hold, then given $r_{\st 1}, C, g^{\st r_2}$, and $\frac{C}{g^{\st r_2}}$, a PPT distinguisher cannot distinguish (i) $g^{\st y_0}$ and $g^{\st y_1}$ form  random elements of $\mathbb{G}$ and (ii) $ \h(\beta_{\st 0}^{\st y_0}) \oplus m_{\st 0}$ and $ \h(\beta_{\st 1}^{\st y_1}) \oplus m_{\st 1}$  from random elements from $\{0, 1\}^{\st\sigma}$, except for a negligible probability  $\mu(\lambda)$.



\end{lemma}

%



\begin{proof}
First, we focus on the first element of   pairs $(g^{\st y_0}, \h(\beta_{\st 0}^{\st y_0}) \oplus m_{\st 0})$ and $(g^{\st y_1}, \h(\beta_{\st 1}^{\st y_1}) \oplus m_{\st 1})$. Since  $y_{\st 0}$ and $y_{\st 1}$ have been picked uniformly at random and unknown to the adversary, $g^{\st y_0}$ and $g^{\st y_1}$ are indistinguishable from random elements of group $\mathbb{G}$. 

Next,  we turn our attention to the second element of the pairs.  Given $C=g^{\st a}$, due to DL problem, value $a$ cannot be extracted by a PPT adversary, except for a probability at most $\mu(\lambda)$. We also know that, due to CDH assumption, a PPT adversary cannot compute  $\beta^{\st y_{i}}_{\st i}$ (i.e., the input of $\h(.)$), given $g^{\st y_{\st i}}, r_1, C, g^{\st r_2}$, and $\frac{C}{g^{\st r_2}}$, where $i\in \{0, 1\}$, except for a probability at most $\mu(\lambda)$. 

We know that $\h(.)$ has been considered as a random oracle and its output is indistinguishable from a random value. Therefore,  $\h(\beta^{\st y_{0}}_{\st 0})\oplus m_{\st 0}$ and  $\h(\beta^{\st y_1}_{\st 1})\oplus m_{\st 1}$ are indistinguishable from random elements of $\{0,1\}^{\st\delta}$, except for a negligible probability, $\mu(\lambda)$. 
\hfill\(\Box\) 
\end{proof}

The main idea behind the design of \rdqothf is as follows. Given a message pair from $\p_{\st 1}$,  \se needs to compute the response for all of the receivers and sends the result to  $\p_{\st 1}$, which picks and sends only one pair in the response to the specific receiver who sent the query and discards the rest of the pairs it received from \se. Therefore, \re receives a single pair (so it cannot learn the total number of receivers or the database size), and the server cannot know which receiver sent the query as it generates the response for all of them. As we will prove, $\p_{\st 1}$ itself cannot learn the actual query of \re, too.

Consider the case where one of the receivers, say  $\re_{\st j}$, wants to send a query. In this case, within \rdqothf, messages $(s_{\st 1}, r_{\st 1})$, $(s_{\st 2}, r_{\st 2})$ and $(\beta_{\st 0}, \beta_{\st 1})$ are generated the same way as they are computed in DQ-OT.  However, given $(\beta_{\st 0}, \beta_{\st 1})$, \se generates $z$ pairs and sends them to $\p_{\st 1}$ who forwards only $v$-th pair to $\re_{\st j}$ and discards the rest.  Given the pair, $\re_{\st j}$ computes the result the same way a receiver does in DQ-OT. Figure \ref{fig::DQHT-OT} in Appendix \ref{sec::DQ-HF-OT-detailed-protocol} presents \rdqothf in detail.


\subsection{Delegated-Unknown-Query Multi-Receiver OT}\label{sec::Delegated-Unknown-Query-OT-HF}

The second variant    \duqothf can be considered as a variant of  \duqot. It is suitable for the setting where servers $\p_{\st 1}$ and $\p_{\st 2}$ do not (and must not) know a client's related index in the sender's database (as well as the index $s$ of the message that the client is interested in).

\subsubsection{Security Definition.}\label{sec::DUQOT-HD-def}

The functionality that \duqothf computes takes as input (i)  a vector of messages $\bm{m}=[(m_{\st 0, 0},m_{\st 1, 0}),\ldots,$ $ (m_{\st 0, z-1},$ $m_{\st 1, z-1})]$ from \se, (ii) an index $v$ of a pair in $\bm{m}$ from $\tp$, (iii) the index $s$ of a message in a pair (where $s\in \{0, 1\}$) from \tp, (iv) the total number $z$ of message pairs from $\tp$, (v)  empty string  \empt from $\p_{\st 1}$, (vi)   \empt from $\p_{\st 2}$, and (vii)  \empt from $\re$. It outputs an empty string $\empt$ to  \se and $\tp$, $z$ to $\p_{\st 1}$, $\empt$ to $\p_{\st 2}$, and outputs to \re $s$-th message from $v$-th pair in $\bm{m}$, i.e., $m_{\st s, v}$. Formally, we define the functionality as: $\mathcal{F}_{\scriptscriptstyle\duqothf}:\big([(m_{\st 0, 0},m_{\st 1, 0}),\ldots,$ $ (m_{\st 0, z-1},$ $m_{\st 1, z-1})], (v, s, z), \empt, \empt, \empt\big) \rightarrow ( \empt, \empt, z, \empt, m_{\st s, v})$, where $v\in\{0,\ldots, z-1\}$. Next, we present a formal definition of \duqothf.


\begin{definition}[\duqothf]\label{def::DUQ-OT-HF-sec-def} Let $\mathcal{F}_{\scriptscriptstyle\duqothf}$ be the functionality defined above. We assert that protocol $\Gamma$ realizes $\mathcal{F}_{\scriptscriptstyle\duqothf}$ in the presence of passive adversary \se, \re, \tp, $\p_{\st 1}$, or $\p_{\st 2}$, if for  every non-uniform PPT adversary \adv
in the real model, there exists a non-uniform PPT simulator \simm  in
the ideal model, such that:

\begin{equation}\label{equ::DUQ-OT-HF-sender-sim-}
\begin{split}
\Big\{\simm_{\st\se}\big(\bm{m}, \empt\big)\Big\}_{\st \bm{m}, s}\stackrel{c}{\equiv}  \Big\{\view_{\st\se}^{\st \Gamma}\big(\bm{m}, (v, s, z), \empt,  \empt, \empt \big) \Big\}_{\st \bm{m}, s}
\end{split}
\end{equation}


\begin{equation}\label{equ::DUQ-OT-HF-server-sim-}
\begin{split}
\Big\{\simm_{\st\p_i}(\empt, out_{\st i})\Big\}_{\st \bm{m}, s}\stackrel{c}{\equiv}  \Big\{\view_{\st\p_i}^{ \st\Gamma}\big(\bm{m}, (v, s, z), \empt,  \empt, \empt \big) \Big\}_{\st \bm{m}, s}
\end{split}
\end{equation}

\begin{equation}\label{equ::DUQ-OT-HF-t-sim-}
\begin{split}
\Big\{\simm_{\st\tp}\big((v, s, z), \empt\big)\Big\}_{\st \bm{m}, s}\stackrel{c}{\equiv}  \Big\{\view_{\st\tp}^{ \st\Gamma}\big(\bm{m}, (v, s, z), \empt,  \empt, \empt \big) \Big\}_{\st \bm{m}, s}
\end{split}
\end{equation}

\begin{equation}\label{equ::DUQ-OT-HF-reciever-sim-}
\begin{split}
\Big\{\simm_{\st\re}\Big(\empt, \mathcal{F}_{\scriptscriptstyle\duqothf}\big(\bm{m}, (v, s, z), \empt,  \empt, \empt \big)\big)\Big)\Big\}_{\st \bm{m}, s}\stackrel{c}{\equiv}  \Big\{\view_{\st\re}^{ \st\Gamma}\big(\bm{m}, (v, s, z), \empt,  \empt, \empt \big) \Big\}_{\st \bm{m}, s}
\end{split}
\end{equation}
where $\bm{m}=[(m_{\st 0, 0},m_{\st 1, 0}),\ldots, (m_{\st 0, z-1}, m_{\st 1, z-1})]$, $out_{\st 1}=z$, $out_{\st 2}=\empt$, and $\forall i$,  $i\in \{1,2\}$. 
\end{definition}

\subsubsection{Protocol.}\label{sec::DUQ-OT-HF} We proceed to present \rduqothf that realizes \duqothf. We build \rdqothf upon protocol DUQ-OT (presented in Figure \ref{fig::DQ-OT-with-unknown-query}). \rduqothf mainly relies on Lemma \ref{lemma::two-pairs-indis-} and the following technique. 
To fetch a record $m_{\st v}$ ``securely'' from a semi-honest \se that holds a database of the form $\bm{a}=[m_{\st 0}, m_{\st 1}, \ldots, m_{\st z-1}]^{\st T}$ where $T$ denotes transpose, 
%
%
without revealing which plaintext record we want to fetch, we can perform as follows: 

\begin{enumerate}
\item set vector $\bm{b}=[b_{\st 0},\ldots, b_{\st z-1}]$, where all $b_{\st i}$s are set to zero except for $v$-th element $b_{\st v}$ which is set to $1$.

\item encrypt each element of $\bm{b}$ using additively homomorphic encryption, e.g., Paillier encryption. Let $\bm{b}'$ be the vector of the encrypted elements. 

 \item send $\bm{b}'$ to the database holder which performs $\bm{b}'\times \bm{a}$ homomorphically, and sends us the single result $res$.
 
 \item decrypt $res$ to discover $m_{\st v}$.\footnote{Such a technique was previously used by Devet \textit{et al.} \cite{DevetGH12} in the ``private information retrieval'' research line.} 
 \end{enumerate}

 In our \rduqothf,  $\bm{b}'$ is not sent for each query to  \se. Instead,  $\bm{b}'$ is stored once in one of the servers, for example, $\p_{\st 1}$. Any time \se computes a vector of responses, say $\bm{a}$, to an OT query, it sends $\bm{a}$ to  $\p_{\st 1}$ which computes $\bm{b}'\times \bm{a}$ homomorphically and sends the result to \re.  Subsequently, \re can decrypt it and find the message it was interested.   Thus, $\p_{\st 1}$ \emph{obliviously filters out} all other records of field elements that do not belong to $\re_j$ and sends to $\re_{\st j}$ only the messages that $\re_{\st j}$ is allowed to fetch. 
 Figure  \ref{fig::DQ-OT-with-unknown-query-and-HF-first-eight-phases} presents \rduqothf in detail.


%
\begin{figure}[!htbp]
\setlength{\fboxsep}{.9pt}
\begin{center}
    \begin{tcolorbox}[enhanced,height=218.2mm,
    colframe=black,colback=white]
\vspace{-2.2mm}

\begin{enumerate}

\item\label{phase::s-init} \underline{\textit{$\se$-side Initialization:}} 
$\mathtt{Init}(1^{\st \lambda})\rightarrow pk$

Chooses a large random prime number $p$, random element
$C \stackrel{\st \$}\leftarrow \mathbb{Z}_p$, and generator $g$. Publishes $pk=(C, p, g)$.

\item\label{DUQOT-HT::gen-key} \underline{\textit{$\re_{\st j}$-side One-off Setup:}}
$\mathtt{\re.Setup}(1^{\st \lambda})\rightarrow (pk_{\st j}, sk_{\st j})$
%


Generates a key pair for the  homomorphic encryption, by calling $\keygen(1^{\st\lambda})\rightarrow(sk_{\st j}, pk_{\st j})$. Sends $pk_{\st j}$  to \tp and $\se$.

\item  \underline{\textit{\tp-side One-off Setup:}}
$\mathtt{\tp.Setup}(z, pk_{\st j})\rightarrow \bm{w}_{\st j}$

\begin{enumerate}

\item initializes an empty vector $\bm{w}_{\st j}=[]$ of size $z$.

\item creates a compressing vector, by setting $v$-th position of $\bm{w}_{\st j}$ to encrypted $1$ and setting the rest of $z-1$ positions to encrypted $0$. $ \forall t, 0\leq t\leq z-1:$

\begin{enumerate}

\item  sets  $d=1$, if $t=v$; sets $d=0$, otherwise. 

\item appends $\enc(pk_{\st j}, d)$ to  $\bm{w}_{\st j}$.

\end{enumerate}

\item sends $\bm{w}_{\st j}$ to $\p_{\st 1}$.

\end{enumerate}

\item \underline{\textit{$\re_{\st j}$-side Delegation:}}
$\mathtt{\re.Request}( pk)\hspace{-.6mm}\rightarrow\hspace{-.6mm} req=(req_{\st 1}, req_{\st 2})$

\begin{enumerate} 

\item picks random values: $r_{\st 1}, r_{\st 2} \stackrel{\st\$}\leftarrow\mathbb{Z}_{\st p}$.

\item sends $req_{\st 1} = r_{\st 1}$ to $\p_{\st 1}$ and $req_{\st 2}= r_{\st 2}$ to $\p_{\st 2}$.

\end{enumerate}

\item \underline{\textit{$\tp$-side Query Generation:}}
$\mathtt{\tp.Request}(1^{\st \lambda}, s,$ $ pk)\hspace{-1mm}\rightarrow\hspace{-1mm} (req'_{\st 1}, $ $req'_{\st 2},$ $ sp_{\st \se})$

\begin{enumerate}

\item splits  the private index $s$ into two shares $(s_{\st 1}, s_{\st 2})$ by calling  $\ses(1^{\st \lambda}, s, 2, 2)\rightarrow (s_{\st 1}, s_{\st 2})$.

\item picks a uniformly random value: $r_{\st 3} \stackrel{\st\$}\leftarrow\{0,1\}^{\st\lambda}$.


\item sends $req'_{\st 1} =s_{\st 1}$ to  $\p_{\st 1}$, $req'_{\st 2}=s_{\st 2}$ to  $\p_{\st 2}$. It sends secret parameter $sp_{\st \se}=r_{\st 3}$ to \se and $sp_{\st \re} = (req'_{\st 2}, sp_{\st \se})$ to \re.

\end{enumerate}
\item \underline{\textit{$\p_{\st 2}$-side Query Generation:}}
$\mathtt{\p_{\st 2}.GenQuery}(req_{\st 2}, req'_{\st 2}, pk)\rightarrow q_{\st 2}$

\begin{enumerate}

\item computes queries:
  $\delta_{\st s_2}= g^{\st r_2},\ \  \delta_{\st 1-s_2} = \frac{C}{g^{\st r_{\st 2}}}$.
\item sends $q_{\st 2}=(\delta_{\st 0}, \delta_{\st 1})$ to  $\p_{\st 1}$. 

\end{enumerate}

\item\label{DUQOT-HT::gen-P1-side-q-gen}\underline{\textit{$\p_{\st 1}$-side Query Generation:}}
$\mathtt{\p_{\st 1}.GenQuery}(req_{\st 1},req'_{\st 1}, q_{\st 2},pk)\hspace{-1.1mm}\rightarrow\hspace{-1.1mm} q_{\st 1}$

\begin{enumerate}

\item computes  queries as: 
%
%
$\beta_{\st s_{\st 1}}=\delta_{\st 0}\cdot g^{\st r_1}, \beta_{\st 1-s_1}=\frac{\delta_{\st 1}} {g^{\st r_1}}$ 

\item sends $q_{\st 1}=(\beta_{\st 0}, \beta_{\st 1})$ to  $\se$.
\end{enumerate}


\item\label{DUQOT-HT::gen-res}\underline{\textit{\se-side Response Generation:}} 
$\mathtt{GenRes}(m_{\st 0,0}, m_{\st 1, 0},\ldots, m_{\st 0,z-1},$ $ m_{\st 1, z-1}, pk, q_{\st 1}, sp_{\st \se})\rightarrow res$

\begin{enumerate}

\item aborts if  $C \neq \beta_{\st 0}\cdot \beta_{\st 1}$.

\item computes a response as follows. $\forall t, 0\leq t \leq z-1:$
\begin{enumerate}
\vspace{-.4mm}
\item   picks two random values $y_{\st 0, t}, y_{\st 1, t}  \stackrel{\$}\leftarrow\mathbb{Z}_{\st p}$.  
  
\item  computes  response:

 $ e_{\st 0, t} := (e_{\st 0, 0, t}, e_{\st 0, 1, t}) = (g^{\st y_{\st 0, t}}, \g(\beta_{\st 0}^{\st y_{0, t}}) \oplus (m_{\st 0, t}||r_{\st 3}))$
      
$e_{\st 1, t} := (e_{\st 1, 0, t}, e_{\st 1, 1, t}) = (g^{\st y_{1, t}}, \g(\beta_{\st 1}^{\st y_{\st 1, t}}) \oplus (m_{\st 1, t}||r_{\st 3})) $ 
   
\item   randomly permutes the elements of each pair $(e_{\st 0, t}, e_{\st 1, t})$ as $\pi(e_{\st 0, t}, e_{\st 1, t})\rightarrow ({e}'_{\st 0, t}, {e}'_{\st 1, t})$.
   
\end{enumerate}

\item     sends $res=(e'_{\st 0, 0}, e'_{\st 1, 0}),\ldots, (e'_{\st 0, z-1}, e'_{\st 1, z-1}) \text{ to } \p_{\st 1}$.
\end{enumerate}

\item\label{DUQOT-HT::oblivius-filter}\underline{\textit{$\p_{\st 1}$-side Oblivious Filtering:}} 
$\mathtt{OblFilter}(res, pk_{\st j}, \bm{w}_{\st j})\hspace{-.7mm}\rightarrow\hspace{-.6mm} res'$

\begin{enumerate}

\item compresses \se's response using vector  $\bm{w}_{\st j}$ as follows. $\forall i,i', 0\leq i,i'\leq 1:$
    
$o_{\st i, i'}= (e'_{\st i, i',0}\hmul \bm{w}_{\st j}[0])\hadd...\hadd  (e'_{\st i,i', z-1}\hmul \bm{w}_{\st j}[z-1])$.
 
\item  sends $res'=(o_{\st 0, 0}, o_{\st 0, 1}), (o_{\st 1, 0}, o_{\st 1, 1}) \text{ to } \re_{\st j}$.

\end{enumerate}

\item\label{DUQOT-HT::message-ext} \underline{\textit{\re-side Message Extraction:}}  
$\mathtt{Retreive}(res', req,  sk_{\st j}, pk,$ $ sp_{\st \re})\hspace{-.6mm}\rightarrow\hspace{-.8mm} m_{\st s}$

\begin{enumerate}

\item decrypts the response from $\p_{\st 1}$ as follows. 
 $\forall i,i', 0\leq i,i'\leq 1:$
$ \dec(sk_{\st j}, o_{\st i,i'})\rightarrow o'_{\st i,i'}$.

\item sets $x=r_{\st 2}+r_{\st 1}\cdot(-1)^{\st s_2}$.

\item retrieves message $m_{\st s, v}$ as follows. 
 $\forall i, 0\leq i\leq1:$

\begin{enumerate}

\item sets $y=\g(({o}'_{\st i, 0})^{\st x})\oplus {o}'_{\st i, 1}$.

\item calls $\parse(\gamma, y)\rightarrow (\ux, \uy)$.

\item  sets $m_{\st s, v}=\ux$, if $\uy=r_{\st 3}$.

\end{enumerate}

\end{enumerate}



    
 















\end{enumerate}
\vspace{-1mm}
\end{tcolorbox}
\end{center}
\vspace{-6mm}
    \caption{Phases \ref{phase::s-init}--\ref{DUQOT-HT::gen-res} of \rduqothf. 
    }
    \label{fig::DQ-OT-with-unknown-query-and-HF-first-eight-phases}
    \vspace{-2mm}
\end{figure}

\begin{theorem}\label{theo::DUQ-OTHF-2-sec}
Let $\mathcal{F}_{\scriptscriptstyle\duqothf}$ be the functionality defined in Section \ref{sec::DUQOT-HD-def}. If  
DL, CDH, and RO assumptions hold and additive homomorphic encryption satisfies IND-CPA, then \rduqothf (presented in Figure \ref{fig::DQ-OT-with-unknown-query-and-HF-first-eight-phases}) securely computes $\mathcal{F}_{\scriptscriptstyle\duqothf}$ in the presence of semi-honest adversaries, 
%
%
w.r.t. Definition \ref{def::DUQ-OT-HF-sec-def}. 
\end{theorem}


\subsection{\rduqothf's Security Proof}\label{sec::proof-of-DUQ-OT-HF}


We prove the security of  \rduqothf, i.e., Theorem \ref{theo::DUQ-OTHF-2-sec}.

\begin{proof}
To prove the theorem, we consider the cases where each party is corrupt at a time.

\subsubsection{Corrupt \re.} In the real execution, \re's view is:  
$\view_{\re}^{\st \rduqothf}\big(\bm{m}, (v, s, z)$ $  \empt, \empt, \empt\big) = \{r_{\st\re}, g, C, p, r_{\st 3}, s_{\st 2}, {o}_{\st 0}, {o}_{\st 1}, $ $m_{\st s, v} \}$, where 
$g$ is a random generator, $p$ is a large random prime number, 
${o}_{\st 0}:=({o}_{\st 0, 0}, {o}_{\st 0, 1})$, ${o}_{\st 1}:=({o}_{\st 1, 0}, {o}_{\st 1, 1})$, 
$C=g^{\st a}$ is a random value and public parameter, $a$ is a random value, and $r_{\st \re}$ is the outcome of the internal random coin of  \re that is used to (i) generate $(r_{\st 1}, r_{\st 2})$ and (ii) its public and private keys pair for additive homomorphic encryption.

We will construct a simulator $\simm_{\st \re}$ that creates a view for \re such that (i) \re will see only a pair of messages rather than $z$ pairs, and (ii) the view is indistinguishable from the view of corrupt \re in the real model. 
$\simm_{\st \re}$ which receives $m_{\st s, v}$ from \re performs as follows.

 \begin{enumerate}
 \item initiates an empty view and appends uniformly random coin $r'_{\st\re}$ to it, where $r'_{\st\re}$ will be used to generate \re-side randomness. It selects a large random prime number $p$ and a random generator $g$.

  \item sets response  as follows: 
  
  \begin{itemize}
 \item picks random values: $C', r'_{\st 1}, r'_{\st 2},y'_{\st 0}, y'_{\st 1}\stackrel{\st\$}\leftarrow\mathbb{Z}_{p}$, $ r'_{\st 3}\stackrel{\st\$}\leftarrow\{0, 1\}^{\st\lambda}$, $s'\stackrel{\st\$}\leftarrow\{0,1\}$, and $u\stackrel{\st\$}\leftarrow\{0, 1\}^{\st\sigma+\lambda}$. 
 
 \item sets $x=r'_{\st 2}+r'_{\st 1}\cdot (-1)^{\st s'}$ and $\beta'_{\st 0}=g^{\st x}$. 
 \item sets $\bar{e}_{\st 0} := \big(\bar{e}_{\st 0,0}=g^{\st y'_0}, \bar{e}_{\st 0, 1}=\g(\beta'^{\st y'_0}_{\st 0})\oplus (m_{\st s,v}|| r'_{\st 3})\big)$ and $\bar{e}_{\st 1} := (\bar{e}_{\st 1,0}=g^{\st y'_{\st 1}}, \bar{e}_{\st 1,1}=u)$. 

 \item encrypts the elements of the pair under $pk$ as follows. $\forall i,i', 0\leq i, i'\leq 1: \bar{o}_{\st i,i'}=\enc(pk, \bar{e}_{\st i,i'})$. Let $\bar{o}_{\st 0}:=(\bar{o}_{\st 0, 0}, \bar{o}_{\st 0, 1})$ and $\bar{o}_{\st 1}:=(\bar{o}_{\st 1, 0}, \bar{o}_{\st 1, 1})$.

 \item randomly permutes the element of  pair $(\bar{o}_{\st 0}, \bar{o}_{\st 1})$. Let $(\bar{o}'_{\st 0}, \bar{o}'_{\st 1})$ be the result.

 \end{itemize}
 \item  appends $(g, C', p, r'_{\st 3}, s', \bar{o}'_{\st 0}, \bar{o}'_{\st 1}, m_{\st s, v})$ to the view and outputs the view.
 
  %

 %
 \end{enumerate}

 Now, we argue that the views in the ideal and real models are indistinguishable. As we are in the semi-honest model, the adversary picks its randomness according to the protocol description; so, $r_{\st \re}$ and $r'_{\st \re}$ model have identical distributions, so do values $(r_{\st 3}, s_{\st 2})$ in the real model and $(r'_{\st 3}, s')$ in the ideal model, component-wise. Moreover,  values $g$ and $p$ in the real and ideal models, as they have been picked uniformly at random.  
For the sake of simplicity, in the ideal model let $\bar{e}'_{\st j}=\bar{e}_{\st 1} = (g^{y'_{\st 1}}, u)$ and in the real model let $e'_{\st i}=e_{\st 1-s}=(g^{\st y_{\st 1-s, v}}, \g(\beta^{\st y_{1-s, v}}_{\st 1-s})\oplus (m_{\st 1-s, v}||r_{\st 3}))$, where $i, j\in\{0,1 \}$.  Note that $\bar{e}'_{\st j}$ and $e'_{\st i}$ contain the elements that the adversary gets after decrypting the messages it receives from $\p_{\st 1}$ in the real model and from $\simm_{\st\re}$ in the ideal model.

We will explain that  $e'_{\st i}$ in the real model and $\bar{e}'_{\st j}$ in the ideal model are indistinguishable. In the real model, it holds that $e_{\st 1-s}=(g^{\st y_{\st 1-s, v}}, \g(\beta^{\st y_{\st 1-s, v}}_{\st 1-s})\oplus (m_{1-s, v}||r_3))$, where $\beta^{\st y_{\st 1-s, v}}_{\st 1-s}=\frac{C}{g^{\st x}}=g^{\st a-x}$. Since $y_{\st 1-s, v}$ in the real model and $y'_{\st 1}$ in the ideal model have been picked uniformly at random and unknown to the adversary, $g^{\st y_{\st 1-s, v}}$ and $g^{\st y'_{1}}$  have identical distributions. Moreover, in the real model, given $C=g^{\st a}$, due to DL problem, $a$ cannot be computed by a PPT adversary.  Also, due to CDH assumption, \re cannot compute  $\beta^{\st y_{1-s, v}}_{\st 1-s}$, given $g^{\st y_{1-s}}$ and $g^{\st a-x}$. We know that $\g(.)$ is considered a random oracle and its output is indistinguishable from a random value. Therefore,  $\g(\beta^{\st y_{1-s, v}}_{\st 1-s})\oplus (m_{\st 1-s, v} || r_{\st 3})$ in the real model and $u$ in the ideal model are indistinguishable. This means that $e'_{\st i}$ and $\bar{e}'_{\st j}$ are indistinguishable too, due to DL, CDH, and RO assumptions. 
 
Also, ciphertexts $\bar{o}_{\st 1,0}=\enc(pk, g^{y'_{\st 1}})$ and $\bar{o}_{\st 1,1}=\enc(pk, u)$ in the ideal model and ciphertexts ${o}_{\st 1-s,0}=\enc(pk, g^{\st y_{1-s, v}})$ and ${o}_{\st 1-s,1}=\enc(pk, \g(\beta^{\st y_{1-s, v}}_{1-s})\oplus (m_{\st 1-s, v}||r_{\st 3}))$ in the real model have identical distributions due to IND-CPA property of the additive homomorphic encryption. 
 Furthermore,  (i) $y_{\st s, v}$ in the real model and $y'_0$ in the ideal model have been picked uniformly at random and (ii) the decryption of both $e'_{\st 1-i}$ and $\bar{e}'_{\st 1-j}$ contain $m_{\st s, v}$; therefore, $e'_{\st 1-i}$ and $\bar{e}'_{\st 1-j}$ have identical distributions. Also, $m_{\st s, v}$ has identical distribution in both models.  Both $C$ and $C'$ have also been picked uniformly at random from $\mathbb{Z}_{\st p-1}$; therefore, they have identical distributions.

 In the ideal model, $\bar{e}_{\st 0}$ always contains encryption of actual message $m_{\st s, v}$ while $\bar{e}_{\st 1}$ always contains a dummy value $u$. However, in the ideal model the encryption of the elements of pair  $(\bar{e}_{\st 0}, \bar{e}_{\st 1})$ and in the real model the encryption of the elements of pair $(e_{\st 0,v}, e_{\st 1,v})$ have been randomly permuted, which results in $(\bar{o}'_{\st 0}, \bar{o}'_{\st 1})$ and $(o_{\st 0}, o_{\st 1})$ respectively. 
Moreover, ciphertexts $\bar{o}_{\st 0,0}=\enc(pk, g^{\st y'_{\st 0}})$ and $\bar{o}_{\st 0,1}=\enc(pk, $ $\g(\beta'^{\st y'_0}_{\st 0})\oplus (m_{\st s,v}|| r'_{\st 3})))$ in the ideal model and ciphertexts ${o}_{\st s,0}=\enc(pk, g^{\st y_{\st s, v}})$ and ${o}_{\st s,1}=\enc(pk, \g(\beta^{y_{\st s, v}}_{\st s})\oplus (m_{\st s, v}||r_{\st 3}))$ have identical distributions due to IND-CPA property of the additive homomorphic encryption. Thus, the permuted pairs have identical distributions too.

We conclude that the two views are computationally indistinguishable, i.e., Relation \ref{equ::DUQ-OT-HF-reciever-sim-} (in Section \ref{sec::Delegated-Unknown-Query-OT-HF}) holds. That means, even though \se holds $z$ pairs of messages and generates a response for all of them, \re's view is still identical to the case where \se holds only two pairs of messages.

 \subsubsection{Corrupt \se.}   This case is identical to the corrupt \se in the proof of DUQ-OT (in Section \ref{sec::DUQ-OT-Security-Proof}) with a minor difference. Specifically, the real-model view of \se in this case is identical to the real-model view of \se in DUQ-OT. Nevertheless, now $\simm_{\st\se}$ receives a vector $\bm{m}=[(m_{\st 0, 0},m_{\st 1, 0}),...,$ $ (m_{\st 0, z-1},$ $m_{\st 1, z-1})]$ from \se, instead of only a single pair that $\simm_{\st\se}$ receives in the proof of DUQ-OT. $\simm_{\st\se}$ still carries out the same way it does in the corrupt \se case in the proof of DUQ-OT. Therefore, the same argument that we used (in Section \ref{sec::DUQ-OT-Security-Proof}) to argue why real model and ideal model views are indistinguishable (when \se is corrupt), can be used in this case as well.
 
 Therefore, Relation \ref{equ::DUQ-OT-HF-sender-sim-}  (in Section \ref{sec::Delegated-Unknown-Query-OT-HF}) holds.

\subsubsection{Corrupt $\p_{\st 2}$.} This case is identical to the corrupt   $\p_{\st 2}$ case in the proof of DUQ-OT. Thus, Relation \ref{equ::DUQ-OT-HF-server-sim-}  (in Section \ref{sec::Delegated-Unknown-Query-OT-HF}) holds.

\subsubsection{Corrupt $\p_{\st 1}$.} In the real execution, $\p_{1}$'s view is: 
$\view_{\p_1}^{\st \rduqothf}\big(\bm{m}, (v, s, z), $ $ \empt, \empt, \empt \big)=\{g, C, p, s_1, \bm{w}_{\st j}, r_{\st 1}, \delta_{\st 0}, \delta_{\st 1},$\\ $ (e'_{\st 0,0}, e'_{\st 1,0}),$ $ ...,$ $ (e'_{\st 0, z-1},$ $ e'_{\st 1, z-1})\}$. Ideal-model $\simm_{\st\p_{\st 1}}$ operates as follows.

\begin{enumerate}
\item initiates an empty view. It selects a large random prime number $p$ and a random generator $g$.
\item picks two random values $\delta'_{\st 0}, \delta'_{\st 1}\stackrel{\st\$}\leftarrow\mathbb{Z}_{p}$. 

\item constructs an empty vector $\bm{w}'$. It picks $z$ uniformly at random elements $w'_{\st 0},..., w'_{\st z}$  from the encryption (ciphertext) range and inserts the elements into $\bm{w}'$. 

\item picks two uniformly random values $s'_{\st 1}\stackrel{\st\$}\leftarrow\mathbb U$ and $C', r'_{\st 1}\stackrel{\st\$}\leftarrow\mathbb{Z}_{p}$, where $\mathbb{U}$ is the output range of $\ses(.)$. 

\item picks $z$ pairs of random values as follows  $(a_{\st 0,0}, a_{\st 1,0}),..., (a_{\st 0,z-1},$ $ a_{\st 1,z-1})\stackrel{\st\$}\leftarrow\mathbb{Z}_{\st p}$. 
\item appends  $s'_{\st 1},g, C', p, r'_{\st 1}, \delta'_{\st 0}, \delta'_{\st 1}$ and pairs $(a_{\st 0,0}, a_{\st 1,0}),..., (a_{\st 0,z-1},$ $ a_{\st 1,z-1})$ to the view and outputs the view. 
\end{enumerate}

Next, we argue that the views in the ideal and real models are indistinguishable. The main difference between this case and the corrupt $\p_{\st 1}$ case in the proof of DUQ-OT (in Section  \ref{sec::DUQ-OT-Security-Proof}) is that now, in the real model, $\p_1$ has: (i) a vector $\bm{w}_{\st j}$ of ciphertexts and (ii) $z$ pairs $(e'_{\st 0,0}, e'_{\st 1,0}),..., (e'_{\st 0,z-1},$ $ e'_{\st 1,z-1})$.  Therefore, we can reuse the same argument we provided for the corrupt $\p_1$ case in the proof of DUQ-OT to argue that the views (excluding $\bm{w}_{\st j}$ and $(e'_{\st 0,0}, e'_{\st 1,0}),..., (e'_{\st 0,z-1},$ $ e'_{\st 1,z-1})$) have identical distributions.

Due to Lemma \ref{lemma::two-pairs-indis-}, the elements of each pair  $(e'_{\st 0, i}, e'_{1, i})$ in the real model are indistinguishable from the elements of each pair  $(a_{\st 0, i}, a_{\st 1, i})$ in the ideal model, for all $i$, $0 \leq i \leq z-1$.  Also, due to the IND-CPA property of the additive homomorphic encryption scheme, the elements of $\bm{w}_{\st j}$ in the real model are indistinguishable from the elements of $\bm{w}'$ in the ideal model.

Hence, Relation \ref{equ::DUQ-OT-HF-server-sim-} (in Section \ref{sec::Delegated-Unknown-Query-OT-HF}) holds.

\subsubsection{Corrupt $\tp$.} This case is identical to the corrupt \tp in the proof of DUQ-OT, with a minor difference; namely, in this case, \tp also has input $z$ which is the total number of message pairs that \se holds. Thus, we can reuse the same argument provided for the corrupt \tp in the proof of DUQ-OT to show that the real and ideal models are indistinguishable. Thus, Relation \ref{equ::DUQ-OT-HF-t-sim-} (in Section \ref{sec::Delegated-Unknown-Query-OT-HF}) holds.
\hfill\(\Box\) 
\end{proof}


\section{A Compiler for Generic OT with Constant Size Response}\label{sec::the-compiler}

In this section, we present a compiler that transforms  \emph{any} $1$-out-of-$n$ OT that requires \re to receive $n$ messages (as a response) into a $1$-out-of-$n$ OT that enables \re to receive only a \emph{constant} number of messages.

 The main technique we rely on is the encrypted binary vector that we used in Section \ref{sec::Delegated-Unknown-Query-OT-HF}.  The high-level idea is as follows.  During query computation, \re (along with its vector that encodes its index $s\in\{0, n-1\}$) computes a binary vector of size $n$, where all elements of the vector are set to $0$ except for $s$-th element which is set to $1$. \re encrypts each element of the vector and sends the result as well as its query to \se. Subsequently, \se computes a response vector (the same manner it does in regular OT), homomorphically multiplies each element of the response by the element of the encrypted vector (component-wise), and then homomorphically sums all the products. It sends the result (which is now constant with regard to $n$) to \re, which  decrypts the response and retrieves the result $m_{\st s}$. 
 
 Next, we will present a generic OT's syntax, and  introduce the generic compiler using the syntax. 
 

\subsection{Syntax of a Conventional OT}\label{sec::OT-syntax}


Since we aim to treat any  OT in a block-box manner, we first present the syntax of an OT. A conventional (or non-delegated) $1$-out-of-$n$ OT (\onenot) has the following algorithms:

\begin{itemize}[leftmargin=4.5mm]

\item[$\bullet$] $\mathtt{\se.Init}(1^{\st\lambda})\rightarrow pk$: a probabilistic algorithm run by \se. It takes as input security parameter $1^{\st\lambda}$ and returns a public key $pk$. 




\item[$\bullet$] $\mathtt{\re.GenQuery}(pk, n, s)\rightarrow (q, sp)$: a probabilistic algorithm run by \re. It takes as input $pk$, the total number of messages $n$, and a secret index $s$. It returns a query (vector) $q$ and a secret parameter $sp$.  

 \item[$\bullet$] $\mathtt{\se.GenRes}(m_{\st 0},\ldots,m_{\st n-1}, pk, q)\rightarrow res$: a probabilistic algorithm run by \se. It takes as input $pk$ and $q$. It generates an encoded response (vector) $res$.

  \item[$\bullet$] $\mathtt{\re.Retreive}(res, q, sp, pk, s)\rightarrow m_{\st s}$: a deterministic algorithm run by \se. It takes as input $res$,  $q$, $sp$,  $pk$, and $s$. It returns message $m_{\st s}$. 
 
\end{itemize}








 

The functionality that a $1$-out-of-$n$ OT  computes can be defined as: 
 $\mathcal{F}_{\scriptscriptstyle\onenot}:\big((m_{\st 0}, \ldots, m_{\st n-1}), s\big) \rightarrow (\empt, m_{\st s})$. Informally, the security of $1$-out-of-$n$ OT states that (1) \re's view can be simulated given its input query $s$ and output message $m_{\st s}$ and (2) \se's view can be simulated given its input messages $(m_{\st 0}, \ldots, m_{\st n-1})$. We refer readers to \cite{DBLP:books/cu/Goldreich2004} for further discussion on $1$-out-of-$n$ OT.

\subsection{The Compiler}\label{sec::compiler}
We present the compiler in detail in Figure \ref{fig::generic-short-res-HE}.  We highlight that in the case where each $e_i \in res$ contains more than one value, e.g., $e_{\st i}=[e_{\st 0,i},\ldots, e_{\st w-1,i}]$ (due to a certain protocol design), then each element of $e_i$ is separately multiplied and added by the element of vector $\bm{b}'$, e.g., the $j$-st element of the response is $e_{\st j, 0}\hmul \bm{b}'[0]\hadd\ldots\hadd e_{\st j, n-1}\hmul \bm{b}'[n-1]$, for all $j$, $0\leq j\leq w-1$. In this case, only $w$ elements are sent to \re.

\begin{figure}[!htbp]
\setlength{\fboxsep}{1pt}
\begin{center}
    \begin{tcolorbox}[enhanced,
    drop fuzzy shadow southwest,
    colframe=black,colback=white]
\begin{enumerate}
\item\underline{\textit{$\se$-side Initialization:}}
$\mathtt{Init}(1^{\st\lambda})\rightarrow pk$

This phase involves \se. 
\begin{enumerate}
    \item calls $\mathtt{\se.Init}(1^{\st\lambda})\rightarrow pk$. 
    \item publishes $pk$.
\end{enumerate}

\item\underline{\re-side Setup:} 
$\mathtt{Setup}(1^{\st\lambda})\rightarrow (sk_{\st \re}, pk_{\st \re})$ 

This phase involves \re. 
\begin{enumerate}
 \item calls 
 $\keygen(1^{\st\lambda})\rightarrow(sk_{\st \re}, pk_{\st \re})$. 

%
\item publishes $pk_{\st \re}$.
 \end{enumerate}
\item\underline{\re-side Query Generation:} 
$\mathtt{GenQuery}(pk_{\st \re}, n, s)\rightarrow (q, sp, \bm{b}')$

This phase involves \re. 
\begin{enumerate}
\item calls $\mathtt{\re.GenQuery}(pk,n, s)\rightarrow (q, sp)$. 
\item constructs a vector $\bm{b}=[b_{\st 0},\ldots,b_{\st n-1}]$, as: 

\begin{enumerate}
\item sets every element $b_{\st i}$ to zero except for $s$-th element $b_{\st s}$ which is set to $1$.
\item encrypts each element of $\bm{b}$ using additive homomorphic encryption, $\forall 0\leq i\leq n-1: b'_{\st i}=\enc(pk_{\st \re}, b_{\st i})$. 
Let $\bm{b}'$ be the vector of the encrypted elements. 
\end{enumerate}
 \item sends $q$ and $\bm{b}'$ to \se and locally stores $sp$. 
\end{enumerate}
\item\underline{\se-side Response Generation:}
$\mathtt{GenRes}(m_{\st 0},\ldots,m_{\st n-1}, pk, pk_{\st \re}, q, \bm{b}')\rightarrow res$

This phase involves \se. 
\begin{enumerate}
\item calls $\mathtt{\se.GenRes}(m_{\st 0},\ldots,m_{\st n-1}, pk, q)\rightarrow res$. Let $res=[e_{\st 0},\ldots, e_{\st n-1}]$. 
\item compresses the response using vector  $\bm{b}'$ as follows. $\forall i, 0\leq i\leq n-1:$ 
    \vspace{-2mm}
$$e= (e_{0}\hmul \bm{b}'[0])\hadd\ldots\hadd  (e_{\st n-1}\hmul \bm{b}'[n-1])$$
 
\item  sends $res=e$  to \re.
\end{enumerate}

\item\underline{\re-side Message Extraction.} 
$\mathtt{Retreive}(res, q, sp, pk,sk_{\st \re}, s)\rightarrow m_{\st s}$

This phase involves \re. 

\begin{enumerate}

\item calls $\dec(sk_{\st \re}, res)\rightarrow res'$. 

\item calls $\mathtt{\re.Retreive}(res', q, sp, pk, s)\rightarrow m_{\st s}$. 
\end{enumerate}
\end{enumerate}
\end{tcolorbox}
\end{center}
\vspace{-4mm}
    \caption{A compiler that turns a $1$-out-of-$n$ OT with response size $O(n)$ to  a $1$-out-of-$n$ OT with response size $O(1)$.}
    \label{fig::generic-short-res-HE}
    \vspace{-3mm}
\end{figure}

\begin{theorem}\label{theo::one-out-of-n-OT}
Let $\mathcal{F}_{\scriptscriptstyle\onenot}$ be the functionality defined above. If  
\onenot is secure and additive homomorphic encryption meets IND-CPA, then generic OT with constant size response (presented in Figure \ref{fig::generic-short-res-HE}) (i) securely computes $\mathcal{F}_{\scriptscriptstyle\onenot}$ in the presence of semi-honest adversaries and (ii) offers $O(1)$ response size, w.r.t. the total number of messages $n$. 
\end{theorem}


\subsection{Proof of Theorem \ref{theo::one-out-of-n-OT}}\label{theo::compiler-sec--}
\begin{proof}[sketch]
Compared to an original $1$-out-of-$n$ OT, the only extra information that \se learns in the real model is a vector of $n$ encrypted binary elements. Since the elements have been encrypted and the encryption satisfies IND-CPA, each ciphertext in the vector is indistinguishable from an element picked uniformly at random from the ciphertext (or encryption) range. Therefore, it would suffice for a simulator to pick $n$ random values and add them to the view. As long as the view of \se in the original $1$-out-of-$n$ OT can be simulated, the view of \se in the new $1$-out-of-$n$ OT can be simulated too (given the above changes). 

Interestingly, in the real model, \re learns less information than it learns in the original $1$-out-of-$n$ OT because it only learns the encryption of the final message $m_s$. The simulator (given $m_s$ and $s$) encrypts $m_s$ the same way as it does in the ideal model in the $1$-out-of-$n$ OT. After that, it encrypts the result again (using the additive homomorphic encryption) and sends the ciphertext to \re. Since in both models, \re receives the same number of values in response, the values have been encrypted twice, and \re can decrypt them using the same approaches, the two models have identical distributions. 

Moreover, the response size is $O(1)$, because the response is the result of (1) multiplying two vectors of size $n$ component-wise and (2) then summing up the products which results in a single value in the case where each element of the response contains a single value (or $w$ values if each element of the response contains $w$ values). 
\hfill\(\Box\) 
\end{proof}

\section{Conclusion}

OT is a crucial privacy-preserving technology.  OTs have found extensive applications in designing secure Multi-Party Computation (MPC) protocols \cite{Yao82b}, Federated Learning (FL) \cite{YangLCT19}, and in accessing sensitive field elements of remote private databases while preserving privacy \cite{CamenischDN09}. In this work, we have identified various gaps both in the privacy of databases and in the OT research. We proposed several novel OTs to address these gaps. 
We have presented a few real-world applications of the proposed OTs, while also formally defining and proving their security within the simulation-based model.

\bibliographystyle{splncs04}
\bibliography{ref}

\appendix




\section{The Original OT of Naor and Pinkas}\label{sec::OT-of-Naor-Pinkas}

Figure \ref{fig::Noar-OT} restates the original OT of Naor and Pinkas \cite[pp. 450, 451]{Efficient-OT-Naor}. 


\begin{figure}[!h]
\setlength{\fboxsep}{1pt}
\begin{center}
    \begin{tcolorbox}[enhanced,
    drop fuzzy shadow southwest,
    colframe=black,colback=white]
\vspace{-2mm}
\begin{enumerate}

\item \underline{\textit{$\se$-side Initialization:}} 
$\mathtt{\se.Init}(1^{\st\lambda})\rightarrow pk$

\begin{enumerate}

\item  chooses a random large prime number $p$ (where $\log_{\st 2}p=\lambda$), a random element
$C \stackrel{\st \$}\leftarrow \mathbb{Z}_{\st p}$ and random generator $g$.
\item publishes $pk=(C, g, p)$. 

\end{enumerate}

\item \underline{\textit{$\re$-side Query Generation:}} 
$\mathtt{\re.GenQuery}(pk, s)\rightarrow (q, sp)$

\begin{enumerate}

\item picks a random value $r \stackrel{\st \$}\leftarrow \mathbb{Z}_p\setminus\{0\}$ and sets $sp=r$. 
\item sets $\beta_{\st s}=g^{\st r}$ and $\beta_{\st 1-s}=\frac{C}{\beta_{\st s}}$.

\item sends $q=\beta_{\st 0}$ to \se and locally stores $sp$.

\end{enumerate}

\item \underline{\textit{$\se$-side Response Generation:}} 
$\mathtt{\se.GenRes}(m_{\st 0}, m_{\st 1}, pk, q)\rightarrow res$

\begin{enumerate}

\item computes $\beta_{\st 1}=\frac{C}{\beta_{\st 0}}$.

\item chooses two random values, $y_{\st 0}, y_{\st 1}\stackrel{\st \$}\leftarrow \mathbb{Z}_p$. 

\item encrypts the elements of the pair $(m_{\st 0}, m_{\st 1})$ as follows: 
$$e_{\st 0}:=(e_{\st 0,0}, e_{\st 0,1})=(g^{\st y_0}, \h(\beta_{\st 0}^{\st y_0}) \oplus m_{\st 0})$$
$$e_{\st 1}:=(e_{\st 1,0}, e_{\st 1,1})=(g^{\st y_1}, \h(\beta_{\st 1}^{\st y_1}) \oplus m_{\st 1})$$

\item sends $res=(e_{\st 0}, e_{\st 1})$ to \re.
\end{enumerate}

\item\underline{\textit{\re-side Message Extraction:}} 
$\mathtt{\re.Retreive}(res, sp, pk, s)\rightarrow m_{\st s}$

\begin{itemize}
\item retrieves the related message $m_{\st s}$ by computing:   
$m_{\st s}=\h((e_{\st s,0})^{\st r})\oplus e_{\st s, 1}$

\end{itemize}

\end{enumerate}
\end{tcolorbox}
\end{center}
\caption{Original OT proposed by  Naor and Pinkas \cite[pp. 450, 451]{Efficient-OT-Naor}. In this protocol, the input of \re is a private binary index $s$ and the input of \se is a pair of private messages $(m_{\st 0}, m_{\st 1})$.} 
\label{fig::Noar-OT}
\end{figure}

\section{\rdqothf in more Detail}\label{sec::DQ-HF-OT-detailed-protocol}
Figure \ref{fig::DQHT-OT} presents the \rdqothf that realizes \dqothf.


\begin{figure}[!htbp]
\setlength{\fboxsep}{1pt}
\begin{center}
    \begin{tcolorbox}[enhanced,
    drop fuzzy shadow southwest,
    colframe=black,colback=white]
\begin{enumerate}
\item \underline{\textit{$\se$-side Initialization:}} 
$\mathtt{Init}(1^{\st \lambda})\rightarrow pk$
\begin{enumerate}

\item chooses a sufficiently large prime number $p$.

\item selects random element
$C \stackrel{\st \$}\leftarrow \mathbb{Z}_p$ and generator $g$.
\item publishes $pk=(C, p, g)$. 

\end{enumerate}

\item \underline{\textit{$\re_j$-side Delegation:}}
$\mathtt{\re.Request}( 1^{\st \lambda}, s, pk)\rightarrow req=(req_{\st 1}, req_{\st 2})$
\begin{enumerate}

\item splits  the private index $s$ into two shares $(s_{\st 1}, s_{\st 2})$ by calling  $\ses(1^{\st \lambda}, s, 2, 2)\rightarrow (s_{\st 1}, s_{\st 2})$.

\item picks two uniformly random values: $r_{\st 1}, r_{\st 2} \stackrel{\st\$}\leftarrow\mathbb{Z}_{\st p}$.

\item sends $req_{\st 1 }=(s_{\st 1}, r_{\st 1})$ to $\p_{\st1}$ and $req_{\st 2 }=(s_{\st 2}, r_{\st 2})$ to $\p_{\st 2}$.

\end{enumerate}

\item \underline{\textit{$\p_{\st 2}$-side Query Generation:}}
$\mathtt{\p_{\st 2}.GenQuery}(req_{\st 2}, ,pk)\rightarrow q_{\st 2}$

\begin{enumerate}

\item computes a pair of partial queries:
\vspace{-2mm}
  $$\delta_{\st s_2}= g^{\st r_2},\ \  \delta_{\st 1-s_2} = \frac{C}{g^{\st r_2}}$$
  
\item sends $q_{\st 2}=(\delta_{\st 0}, \delta_{\st 1})$ to  $\p_{\st 1}$. 

\end{enumerate}

\item\underline{\textit{$\p_{\st 1}$-side Query Generation:}}
$\mathtt{\p_{\st 1}.GenQuery}(req_{\st 1}, q_{\st 2},pk)\rightarrow q_{\st 1}$

\begin{enumerate}

\item computes a pair of final queries as: 
\vspace{-2mm}
$$\beta_{\st s_1}=\delta_{0}\cdot g^{\st r_1},\ \ \beta_{\st 1-s_1}=\frac{\delta_{\st 1}} {g^{\st r_1}}$$

\item sends $q_{\st 1}=(\beta_{\st 0}, \beta_{\st 1})$ to  $\se$.

\end{enumerate}

\item\underline{\textit{\se-side Response Generation:}} 
$\mathtt{GenRes}(m_{\st 0, 0}, m_{\st 1, 0},\ldots, m_{\st 0, z-1}, m_{\st 1, z-1},  pk, q_{\st 1})\rightarrow res$

\begin{enumerate}

\item aborts if  $C \neq \beta_{\st 0}\cdot \beta_{\st 1}$.
\item computes a response as follows. $\forall t, 0\leq t\leq z-1:$

\begin{enumerate}[leftmargin=3.5mm]

\item  picks two random values $y_{\st 0, t}, y_{\st 1, t}  \stackrel{\$}\leftarrow\mathbb{Z}_{\st p}$.
\item computes response:
\vspace{-2mm}
   $$e_{\st 0, t} := (e_{\st 0, 0, t}, e_{\st 0, 1, t}) = (g^{\st y_{0, t}}, \h(\beta_{\st 0}^{\st y_{0, t}}) \oplus m_{\st 0, t})$$
$$e_{\st 1, t} := (e_{\st 1, 0, t}, e_{\st 1, 1, t}) = (g^{\st y_{1, t}}, \h(\beta_{\st 1}^{\st y_{1, t}}) \oplus m_{\st 1, t}) $$

\end{enumerate}

\item sends $res=(e_{\st 0, 0}, e_{\st 1, 0}),..., (e_{\st 0, z-1}, e_{\st 1, z-1})$ to $\p_{\st 1}$. 

\end{enumerate}

\item\underline{\textit{$\p_{\st 1}$-side Oblivious Filtering:}} 
$\mathtt{OblFilter}(res)\rightarrow  res'$

\begin{itemize}
\item[$\bullet$] forwards $res'=(e_{\st 0, v}, e_{\st 1, v})$ to $\re_{\st j}$ and discards the rest of the messages received from \se. 
\end{itemize}

\item\underline{\textit{\re-side Message Extraction:}} 
$\mathtt{Retreive}(res', req,   pk)  \rightarrow m_{\st s}$

\begin{enumerate}

\item sets $x=r_{\st 2}+r_{\st 1}\cdot(-1)^{\st s_2}$.

\item retrieves message $m_{\st s,v}$ by setting:  
 $m_{\st s, v}=\h((e_{\st s, 0, v})^{\st x})\oplus e_{\st s, 1 v}$

\end{enumerate}

\end{enumerate}
\end{tcolorbox}
\end{center}
\vspace{-2mm}
    \caption{\rdqothf: Our protocol that realizes \dqothf.}
    \label{fig::DQHT-OT}
    \vspace{-3mm}
\end{figure}

\begin{theorem}\label{theo::DQ-OTHF-sec}
Let $\mathcal{F}_{\scriptscriptstyle\dqothf}$ be the functionality defined in Section \ref{sec::DQOT-HF}. If  
DL, CDH, and RO assumptions hold, then \rdqothf (presented in Figure \ref{fig::DQHT-OT}) securely computes $\mathcal{F}_{\scriptscriptstyle\dqothf}$ in the presence of semi-honest adversaries, 
%
%
w.r.t. Definition \ref{def::DQ-OT-HF-sec-def}. 
\end{theorem}



\subsection{Proof of  Theorem \ref{theo::DQ-OTHF-sec}}\label{sec::proof-of-DQ-OTHF-sec}
Below, we prove the security of  \rdqothf, i.e., Theorem \ref{theo::DQ-OTHF-sec}.

\begin{proof}
To prove the above theorem, we consider the cases where each party is corrupt at a time.

\subsubsection{Corrupt \re.} Recall that in \rdqothf, sender  \se holds a vector $\bm{m}$ of $z$ pairs of messages (as opposed to DQ-OT where \se holds only a single pair of messages). In the real execution, \re's view is:  
$\view_{\st\re}^{\st \rdqothf}\big(m_{\st 0}, m_{\st 1}, v,$ $  \empt, s\big) = \{r_{\st\re}, g, C, p, e_{\st 0,v}, e_{\st 1, v}, m_{\st s,v} \}$, where $C=g^{\st a}$ is a random value and public parameter, 
where $g$ is a random generator, $a$ is a random value, $p$ is a large random prime number, and $r_{\st \re}$ is the outcome of the internal random coin of  \re and is used to generate $(r_{\st 1}, r_{\st 2})$.

We will construct a simulator $\simm_{\st\re}$ that creates a view for \re such that (i) \re will see only a pair of messages (rather than $z$ pairs), and (ii) the view is indistinguishable from the view of corrupt \re in the real model. 
$\simm_{\st \re}$ which receives $(s, m_{\st s})$ from \re operates as follows.

 \begin{enumerate}
 \item initiates an empty view and appends uniformly random coin $r'_{\st\re}$ to it, where $r'_{\st\re}$ will be used to generate \re-side randomness. It chooses a large random prime number $p$ and a random generator $g$.

 \item sets $(e'_{\st 0}, e'_{\st 1})$ as follows: 
 \begin{itemize}
 \item splits $s$ into two shares: $\ses(1^{\st\lambda}, s, 2, 2)\rightarrow (s'_{\st 1}, s'_{\st 2})$.
 \item picks uniformly random values: $C', r'_{\st 1}, r'_{\st 2},y'_{\st 0}, y'_{\st 1}\stackrel{\st\$}\leftarrow\mathbb{Z}_{\st p}$.
 %
 %
 \item sets $\beta'_s=g^{\st x}$, where $x$ is set as follows: 
 \begin{itemize}
 \item[$*$] $x = r'_{\st 2} + r'_{\st 1}$, if $(s = s_{\st 1} = s_{\st 2} = 0)$ or $(s = s_{\st 1} = 1\wedge s_{\st 2} = 0)$.
 \item[$*$] $x = r'_{\st 2} - r'_{\st 1}$, if $(s =0 \wedge s_{\st 1} = s_{\st 2} = 1)$ or $(s = s_{\st 2} = 1\wedge s_{\st 1} = 0)$.
 \end{itemize}
 \item picks a uniformly random value $u\stackrel{\st\$}\leftarrow\mathbb{Z}_{p}$ and then sets $e'_s = (g^{\st y'_s}, \h(\beta'^{\st y'_s}_{\st s})\oplus m_{\st s})$ and $e'_{\st 1-s} = (g^{\st y'_{1-s}}, u)$. 
 \end{itemize}
 \item  appends $(g, C', p, r'_{\st 1}, r'_{\st 2},  e'_{\st 0}, e'_{\st 1}, m_{\st s})$ to the view and outputs the view.
 \end{enumerate}

The above simulator is identical to the simulator we constructed for DQ-OT. Thus, the same argument that we used (in the corrupt \re case in Section \ref{sec::DQ-OT-proof}) to argue why real model and ideal model views are indistinguishable, can be used in this case as well.  That means, even though \se holds $z$ pairs of messages and generates a response for all of them, \re's view is still identical to the case where \se holds only two pairs of messages. 
Hence, Relation \ref{equ::reciever-sim-DQ-OT-HF}  (in Section \ref{sec::DQOT-HF}) holds.


 \subsubsection{Corrupt \se.}   This case is identical to the corrupt \se in the proof of DQ-OT (in Section \ref{sec::DQ-OT-proof}) with a minor difference. Specifically, the real-model view of \se in this case is identical to the real-model view of \se in DQ-OT; however, now $\simm_{\st\se}$ receives a vector $\bm{m}=[(m_{\st 0, 0},m_{\st 1, 0}),...,$ $ (m_{\st 0, z-1},$ $m_{\st 1, z-1})]$ from \se, instead of only a single pair that $\simm_{\st\se}$ receives in the proof of DQ-OT. $\simm_{\st\se}$ still operates the same way it does in the corrupt \se case in the proof of DQ-OT. Therefore, the same argument that we used (in Section \ref{sec::DQ-OT-proof}) to argue why real model and ideal model views are indistinguishable (when \se is corrupt), can be used in this case as well.
 
 Therefore, Relation \ref{equ::sender-sim-DQ-OT-HF}  (in Section \ref{sec::DQOT-HF}) holds.

\subsubsection{Corrupt $\p_{\st 2}$.} This case is identical to the corrupt   $\p_{\st 2}$ case in the proof of DQ-OT. So, Relation \ref{equ::server2-sim-DQ-OT-HF}  (in Section \ref{sec::DQOT-HF}) holds.

\subsubsection{Corrupt $\p_{\st 1}$.} In the real execution, $\p_{\st 1}$'s view is: 
$\view_{\st\p_1}^{\st \rdqothf}\big((m_{\st 0}, m_{\st 1}), v,$ $ \empt, s\big)=\{g, C, p, s_{\st 1}, r_{\st 1}, \delta_{\st 0}, \delta_{\st 1}, (e_{\st 0,0},$ \\ $ e_{\st 1,0}),$ $\ldots, (e_{\st 0, z-1},$ $ e_{\st 1, z-1})\}$. Ideal-model $\simm_{\st\p_{1}}$ that receives $v$ from $\p_{\st 1}$ operates as follows.

\begin{enumerate}
\item initiates an empty view. It selects a large random prime number $p$ and a random generator $g$.

\item picks two random values $\delta'_{\st 0}, \delta'_{\st 1}\stackrel{\st\$}\leftarrow\mathbb{Z}_{\st p}$. 
\item picks two uniformly random values $s'_{\st 1}\stackrel{\st\$}\leftarrow\mathbb U$ and $C', r'_{\st 1}\stackrel{\st\$}\leftarrow\mathbb{Z}_{\st p-1}$, where $\mathbb{U}$ is the output range of $\ses(.)$. 

\item picks $z$ pairs of random values as follows  $(a_{\st 0,0}, a_{\st 1,0}),..., (a_{\st 0,z-1},$ $ a_{\st 1,z-1})\stackrel{\st\$}\leftarrow\mathbb{Z}_{\st p}$. 
\item appends  $s'_{\st 1}, g, C',p, r'_{\st 1}, \delta'_{\st 0}, \delta'_{\st 1}$ and pairs $(a_{\st 0,0}, a_{\st 1,0}),..., (a_{\st 0,z-1},$ $ a_{\st 1,z-1})$ to the view and outputs the view. 
\end{enumerate}

Now, we explain why the views in the ideal and real models are indistinguishable. The main difference between this case and the corrupt $\p_{\st 1}$ case in the proof of DQ-OT (in Section  \ref{sec::DQ-OT-proof}) is that now $\p_{\st 1}$ has $z$ additional pairs $(e_{\st 0,0}, a_{\st 1,0}),..., (e_{\st 0,z-1},$ $ a_{\st 1,z-1})$. Therefore, regarding the views in real and ideal models excluding the additional $z$ pairs, we can use the same argument we provided for the corrupt $\p_1$ case in the proof of DQ-OT to show that the two views are indistinguishable. Moreover, due to Lemma \ref{lemma::two-pairs-indis-}, the elements of each pair  $(e_{\st 0,i}, e_{\st 1,i})$ in the real model are indistinguishable from the elements of each pair  $(a_{\st 0,i}, a_{\st 1,i})$ in the ideal model, for all $i$, $0 \leq i \leq z-1$.  Hence, Relation \ref{equ::server1-sim-DQ-OT-HF} (in Section \ref{sec::DQOT-HF}) holds.
\hfill\(\Box\) 
\end{proof}

\end{document}